\documentclass[journal]{IEEEtran}
\usepackage{graphicx}
\usepackage{amssymb,amsmath,amsthm}
\usepackage[ruled,vlined]{algorithm2e}
\usepackage{xcolor}
\usepackage{booktabs}
\usepackage{url}
\usepackage[latin1]{inputenc}

\usepackage[caption=false,font=footnotesize]{subfig}

\usepackage{enumitem}

\newcounter{ctheorem}
\newtheorem{theorem}[ctheorem]{Theorem}

\newcounter{clemma}
\newtheorem{lemma}[clemma]{Lemma}
\newcounter{ccorollary}
\newtheorem{corollary}[ccorollary]{Corollary}

\newcommand{\otsikko}{Shrinking the eigenvalues of M-estimators of covariance matrix}

\newcommand{\keywords}{M-estimators, sample covariance matrix, shrinkage, regularization, elliptically symmetric distributions}

\usepackage{hyperref}
\usepackage{tikz}
\usepackage{pgfplots}
\pgfplotsset{compat=1.13}

\usetikzlibrary{arrows,petri,topaths}
\usetikzlibrary{plotmarks}
  \newlength\fwidth

\newcommand{\E}{\mathbb{E}}                  % expectation
\newcommand{\bo}[1]{\mathbf{#1}}              % boldface math 
\newcommand{\bom}[1]{\boldsymbol{#1}}    % boldface math (for greek letters)

\newcommand{\be}{\beta}
  
\newcommand{\ka}{\kappa}

\newcommand{\pdim}{p}
\newcommand{\ndim}{n}
\newcommand{\wins}{\textsf{wins}}

\renewcommand{\Pi}{{\mathcal P}}
\newcommand{\hop}{\mathsf{H}}        % Hermitian transpose 

\newcommand{\beq}{\begin{equation}}
\newcommand{\eeq}{\end{equation}}
\newcommand{\bmat}{\begin{pmatrix}}
\newcommand{\emat}{\end{pmatrix}}

\newcommand{\I}{\bo I}

\newcommand{\X}{{\bf X}}

\newcommand{\A}{{\bo A}}
\newcommand{\B}{{\bo B}}
\renewcommand{\S}{{\bo S}} % SAMPLE COVARIANCE MATRIX
\newcommand{\C}{\mathbb{C}} %  field of complex numbers

\newcommand{\x}{\bo x}

\renewcommand{\u}{\bo v}

\newcommand{\Fr}{\mathrm{F}}

\newcommand{\M}{\bom \Sigma}
\newcommand{\Mn}{\M_0}  	    % population M-estimator 
\newcommand{\Mor}{\bo C}	    % 1step oracle estimaotr 

\newcommand{\MSE}{\mathrm{MSE}}

\newcommand{\sca}{\sigma}

\DeclareMathOperator{\tr}{tr}
\DeclareMathOperator{\Tr}{tr}

\DeclareMathOperator{\cov}{cov}

\begin{document}

\title{\otsikko}
\author{Esa~Ollila,~\IEEEmembership{Member,~IEEE},~Daniel P. Palomar,~\IEEEmembership{Fellow,~IEEE}~and~Fr\'ed\'eric Pascal,~\IEEEmembership{Senior Member,~IEEE}
\thanks{E. Ollila is with the Department of Signal
	Processing and Acoustics, Aalto University, P.O. Box 15400, FI-00076
Aalto, Finland. D. Palomar is with the Hong Kong University of Science and Technology, Hong Kong. F. Pascal is with Universit\'e Paris-Saclay, CNRS, CentraleSup\'elec, Laboratoire des signaux et syst\`emes, 91190, Gif-sur-Yvette, France.}% <-this % stops a space
\thanks{Manuscript received November xx, 202X; revised xx, 202X.}}

\markboth{IEEE TRANSACTIONS ON SIGNAL PROCESSING,~Vol.~xxx, No.~x, 2020}%
{Ollila \MakeLowercase{\textit{et al.}}: \otsikko}




\maketitle

\begin{abstract}
A highly popular regularized (shrinkage) covariance matrix estimator is the shrinkage sample covariance matrix (SCM)  which shares the same set of eigenvectors as the SCM but shrinks its eigenvalues toward the grand mean of the eigenvalues of the SCM. In this paper, a more general approach is considered in which  the SCM is replaced  by an M-estimator of scatter matrix and a fully automatic data adaptive method to compute the optimal shrinkage parameter with minimum mean squared error is proposed. Our approach permits the use of any weight function such as Gaussian, Huber's, Tyler's, or $t$ weight functions, all of which are commonly used in M-estimation framework. Our simulation examples illustrate that shrinkage M-estimators based on the proposed optimal tuning combined with robust weight function do not loose in performance to shrinkage SCM estimator when the data is Gaussian, but provide significantly improved performance when the data is sampled from an unspecified heavy-tailed elliptically symmetric distribution. Also, real-world and synthetic stock market data validate the performance of the proposed method in practical applications. 
\end{abstract}

\begin{IEEEkeywords}
	\keywords
\end{IEEEkeywords}

% \fi
\IEEEpeerreviewmaketitle

%%%%%%%%%%%%%%%%%%%%%%%%%%%%%%%%%%%%%%%%%%%%%%%%%%%%%%%%%%%%%%%%%%%%%%%%%%%%%%%

%%%%%%%%%
% SECTION 1
%%%%%%%%%
\section{Introduction} \label{sec:intro}

Consider a sample of $p$-dimensional vectors $\{ \x_i \}_{i=1}^n$ sampled from a distribution of a random vector $\x$ with mean vector equal to zero (i.e., $\E[\x]= \bo 0$).   One of the first task in the analysis of high-dimensional data is to estimate the covariance matrix.   The most commonly used estimator is the sample covariance matrix (SCM), $\S=\frac 1 n \sum_{i=1}^\ndim \x_i \x_i^\top$, 
but its main drawbacks are  its loss of efficiency when sampling from distributions which have heavier 
tails than the multivariate normal (MVN) distribution and its sensitivity to outliers.  Although being unbiased estimator of 
the covariance matrix $\cov(\x) = \E[\x \x^\top]$ for any sample length $n \geq 1$, it is well-known that the eigenvalues are poorly estimated when $n$ is not orders of magnitude larger than $p$. In such cases, one commonly uses a  \emph{regularized SCM (RSCM)} with a linear shrinkage towards a scaled identity matrix, defined as 
\beq \label{eq:LW}
\S_\be = \be \S + (1-\be) \frac{\Tr(\S)}{p} \I,
\eeq 
where $\be \in [0,1]$ is the regularization parameter.  The RSCM $\S_\be$ shares the same set of eigenvectors as the SCM $\S$, but its eigenvalues are shrinked towards the grand mean of the eigenvalues of the SCM $\S$. That is, if $d_1,\ldots, d_p$ denote the eigenvalues of $\bo S$, then $ \be d_j + (1-\be)\bar d$ are the eigenvalues of  $\bo S_\be$, where $\bar d=p^{-1}\sum_{j=1}^p d_j$. Optimal computation of  $\be$ such that $\S_\be$ has minimum mean squared error (MMSE) has been developed for example in  \cite{ledoit2004well,ollila2019optimal} or in \cite{li2017estimation} for  certain structured target matrices.
 A Bayesian approach has been considered in \cite{coluccia2015regularized}.

However, the estimator in \eqref{eq:LW} remains sensitive to outliers and non-Gaussianity.  M-estimators of scatter \cite{maronna1976robust} are popular robust alternatives to SCM. We consider the situation where  $n>p$ and hence a conventional M-estimator of scatter  $\hat \M$ exists under mild conditions on the data (see \cite{kent1991redescending})  and can thus be used  in place of the SCM $\bo S$  in  \eqref{eq:LW}.  We then propose a fully automatic data adaptive method to  compute the optimal shrinkage parameter $\be$.    First, we derive an approximation for the optimal parameter $\be_o$ attaining the minimum MSE and then propose a data adaptive method for its computation.  The benefit of our approach is that it can be easily applied to any M-estimator using any weight function $u(t)$.   Our simulation examples illustrate that a shrinkage M-estimator using the proposed tuning and a robust weight function does not loose in performance to  optimal shrinkage SCM estimator when the data is Gaussian, but is able to provide significantly improved performance in the case of heavy-tailed data and in the presence of outliers. 

Earlier works, e.g., in \cite{abramovich2007diagonally,chen2010shrinkage,du2010fully,ollila2014regularized,pascal2014generalized, sun2014regularized,couillet2014large,zhang2016automatic,yi2020shrinking},  proposed regularized M-estimators of scatter matrix either  by  adding a penalty function to M-estimation objective function or a diagonal loading term to the respective first-order solution (M-estimating equation).  We consider  a simpler approach that uses a conventional M-estimator and shrinks its  eigenvalues to the grand mean of the eigenvalues.  Our approach permits  computation of the optimal shrinkage parameter for any M-estimation weight function. Preliminary study of the proposed estimators has appeared in a conference proceeding \cite{ollila2020m}. 

Finally, we note some related but different approaches to what is pursed in this paper.  For example, covariance matrix estimation  in the low sample size large dimensionality setting commonly arises in radar signal processing, where often  some constrained,  mismatched or structured estimation framework of the covariance matrix is exploited. See e.g.,   \cite{aubry2013radar,sun2016low,aubry2018geometric,demaio2019loading,tang2020invariance,besson2020maximum}. On the other hand, there are also other approaches for parameter tuning of regularized covariance matrix estimators such as  cross-validation or expected likelihood approach \cite{abramovich2013regularized,besson2013regularized,abramovich2015expected}.

The paper is structured as follows. \autoref{sec:shrinkM} introduces the proposed  shrinkage M-estimator framework while 
\autoref{sec:beta_opt}  discusses automatic computation of the optimal shrinkage parameter under the assumption of sampling from unspecified elliptical distribution.  Extension to complex  case is discussed in \autoref{sec:complex_ext}.   \autoref{sec:special_cases}  addresses the most commonly used M-estimators, namely, the Gaussian weight function, Huber's weight function, Tyler's weight and $t$-distribution weight function. We provide simulation studies in \autoref{sec:simul}  and experimental results in \autoref{sec:applic}.  In  \autoref{sec:applic} we validate the promising performance of the proposed approach 
both in the case of  real-world and synthetic stock market data. Finally, \autoref{sec:concl} concludes, while proofs of theorems and lemmas are kept in the Appendix.

%%%%%%%%%
% SECTION 2
%%%%%%%%%

\section{Shrinkage M-estimators of scatter} \label{sec:shrinkM}

In this paper, we assume that $n>p$ (except for Gaussian loss) and consider an M-estimator of scatter matrix \cite{maronna1976robust} that 
solves an estimating equation
\beq \label{eq:Mest}
\hat \M = \frac{1}{n} \sum_{i=1}^n u(\x_i^\top \hat \M^{-1} \x_i) \x_i \x_i^\top , 
\eeq 
where $u: [0, \infty) \to [0,\infty)$ is a non-increasing \emph{weight function}.  An M-estimator is a sort of adaptively weighted SCM with weights determined by function $u(\cdot)$.   To guarantee existence of the solution,  it is required that the data verifies the condition stated in   \cite{kent1991redescending}. An M-estimator of scatter which shrinks the eigenvalues towards the grand mean of the eigenvalues is then defined as:  
\beq \label{eq:shrinkMest}
\hat \M_\be = \be \hat \M + (1-\be) \frac{\Tr(\hat \M)}{p} \I. 
\eeq 
Thus $\hat \M_\be$ is indexed by the  shrinkage parameter $\beta \in [0,1]$.  If $\beta=1$, then $\hat \M_\be$ coincides with the conventional M-estimator in \eqref{eq:Mest} while if $\beta=0$, then $\hat \M_\be$ equals an identity matrix scaled by mean of the eigenvalues of $\hat \M$. Next we discuss the commonly used weight functions $u$. 

The RSCM $\bo S_\be$ in \eqref{eq:LW} is obtained when one uses the \emph{Gaussian weight function} $u_{\text{G}}(t)=1$ $\forall t$. Terminology `Gaussian' stems from the fact that  $\hat \M = \bo S$ is the  maximum likelihood estimate (MLE)  of the covariance matrix of MVN distribution. We return to this in \autoref{sec:beta_opt}. \emph{Huber's weight function} is defined as 
\beq \label{eq:huber_u}
u_{\mbox{\tiny H}} (t;c) =   \begin{cases}  1/b,  
&  \ \mbox{for} \ t \leqslant  c^2 \\ c^2/(t b),  & \ \mbox{for} \ t > c^2 \end{cases}  
\eeq 
where   $c > 0$ is a user defined tuning constant that determines the robustness and efficiency of the estimator and  
$b$ is  a scaling factor; see \autoref{subsect:huber}  for more details. 
Another popular choice is  \emph{MVT-weight function} \cite{kent1991redescending}:
\beq \label{eq:t_u}
u_{\mbox{\tiny T}}(t ; \nu) = \frac{\pdim  + \nu }{\nu +t}
\eeq 
in which case the corresponding M-estimator $\hat \M$  is also the MLE of the scatter matrix parameter of the \emph{multivariate $t$ (MVT) distribution}  with $\nu>0$ degrees of freedom (d.o.f.). We return to this estimator in  \autoref{subsect:tMLE}. 
Finally, another classic choice, with nice robustness properties,  is Tyler's \cite{tyler1987distribution} M-estimator, in which case the weight function is  
\beq \label{eq:tyl_u} 
u_{\mbox{\tiny Tyl}}(t)= \frac{p}{t}. 
\eeq 
Both Huber's and MVT-weight function yield Tyler's weight function as special cases; namely, 
for $\nu=0$, one notices that $u_{\mbox{\tiny T}}(t ; \nu=0) =u_{\mbox{\tiny Tyl}}(t)$ and in the limit case 
as $c \to 0$   Huber's weight function tends to Tyler's weight function.  

We would like to stress that $n>p$ is a necessary assumption for all but Gaussian weight functions for a solution $\hat \M$ to \eqref{eq:Mest} to exist. We do not include the limit case $n=p$ since any affine equivariant M-estimator $\hat \M$ when $n=p$ and data is in general position is just a scalar multiple of the SCM $\bo S$, that is, $\hat \M = \gamma \bo S$ for some $\gamma>0$ \cite{tyler2010note}.  For example, M-estimator based on Huber's or $t$-weights are affine equivariant. Moreover, note that Theorem 2 in \cite{couillet2015therandom} ensures similar results in the large sample regime.  Namely, the following convergence is proved 
 $$\|\hat \M - \hat{\bo S}_{np}\| \xrightarrow[n, p \to \infty]{a.s} 0$$
 with $n/p\to c \in (0,1)$ and $\hat{\bo S}_{np}$ is an appropriate weighted SCM, and the norm denotes the spectral norm.   

An M-estimator is a consistent estimator of the M-functional of scatter matrix,  defined as 
 \beq \label{eq:Mfun}
\Mn = \E\big[ u(\x^\top \Mn^{-1} \x)\x\x^\top \big].
\eeq  
If the population M-functional $\Mn$ is known,  then by defining a  \emph{1-step estimator}
\beq \label{eq:Mor} 
\Mor =    \frac{1}{n} \sum_{i=1}^n u(\x_i^\top \Mn^{-1} \x_i) \x_i \x_i^\top   , 
\eeq
we can compute 
\beq\label{eq:Mor2}
 \bo C_\be =   \be  \Mor + (1-\be)[\tr(\bo C)/p] \I ,
\eeq 
which serves as  a proxy for $\hat \M_\be$.  Naturally, such an estimator is fictional, since the initial value $\Mn$ is unknown. 
 The 1-step estimator $\Mor$ is, by  its construction, an unbiased estimator of  $\Mn$,  i.e., $\E[\Mor] = \Mn$. 

Ideally we would like to find the value of $\beta  \in [0,1]$ for which  the corresponding estimator $\hat \M_\be$ attains the minimum MSE, that is, 
\beq \label{eq:optimal_beta_o}
\beta_o = \arg \min_\be  \Big\{ \mathrm{MSE}(\hat \M_\be) =  \mathbb{E} \Big[ \big\| \hat \M_\be - \M_0 \|^2_{\mathrm{F}} \Big] \Big\},
\eeq 
where $\| \cdot \|_{\Fr}$ denotes the Frobenius matrix norm (i.e., $\| \A \|_{\Fr}^2 = \tr(\bo A^\top \bo A)$ and $\| \A \|_{\Fr}^2 = \tr(\bo A^\hop \bo A)$ for real-valued and complex-valued matrices, respectively, where $(\cdot)^\hop$ denotes the Hermitian transpose).  
Since solving \eqref{eq:optimal_beta_o} is not feasible due to the implicit form of M-estimators, we instead solve the following  
much simpler problem: 
\beq \label{eq:RFegSCM_oracle} 
\be_{o}^{\text{app}} =   \underset{\be}{\arg \min}  \ \Big \{ \mathrm{MSE}(\bo C_\be) =  \E \Big[ \big\| \bo C_\be - \Mn \big\|^2_{\rm F}  \Big] \Big\} .   
\eeq 
Such approach was also used  in \cite{chen2011robust} to derive an optimal parameter for the shrinkage Tyler's M-estimator of scatter proposed by the authors, 

Before stating the expression for $ \be_{o}^{\text{app}}$ we  introduce a \emph{sphericity}  measure of scatter, defined as 
\beq \label{eq:gamma} 
 \gamma =  \frac{p \tr(\Mn^2)}{\tr(\Mn)^2}   .
\eeq 
Sphericity $\gamma$  \cite{ledoit2002some,srivastava2005some} measures how close $\Mn$ is to a scaled identity matrix: $\gamma \in [1,p]$ 
where $\gamma=1$ if and only if $\Mn \propto \I$ and $\gamma = p$ if $\Mn$ has rank equal to 1.

\begin{theorem}  \label{th:beta0}    
Suppose $\x_1, \ldots, \x_n$  is an i.i.d.\ random sample from any $p$-variate distribution,    and 
$u$ is a weight function for which the expectation $\E[ \tr(\bo C^2)]$ exists. 
The oracle  parameter $\be_o^{\mathrm{app}}$ in \eqref{eq:RFegSCM_oracle} 
 is
\begin{align}
\be_o^{\mathrm{app}}  &= \frac{\| \Mn - \eta_o \I \|_{\Fr}^2}{\E \Big [ \big\| \bo C -   (\tr(\bo C)/p) \I \big\|_{\Fr}^2 \Big]} \label{eq:beta0id} \\ 
	    &= \frac{p (\gamma-1) \eta^2_o}{\E[\tr(\bo C^2)] - p^{-1} \E[ \tr(\bo C)^2]} \label{eq:beta0id2} 
\end{align}
where $\eta_o  = \tr(\Mn)/p$ and $\gamma$ is defined in \eqref{eq:gamma}. Note that $\be_o^\mathrm{app} \in \left[0,1\right)$ and the value of the MSE
at the optimum  is
\begin{equation} \label{eq:MSEopt}
 \MSE(\bo C_{\be_o^{\mathrm{app}}})  =  \frac{ \E[\tr(\bo C)^2 ] - \tr(\Mn)^2}{p}+ ( 1-  \beta_o^{\mathrm{app}})  \big\| \Mn - \eta_o  \bo I  \big \|^2_{\Fr}.
\end{equation}
\end{theorem}

\begin{proof}
The proof is postponed to Appendix~\ref{app:th:beta0}.
\end{proof}

In the next section, we derive a more explicit form of  $\be_o^{\mathrm{app}}$ by assuming that the data is generated from an unspecified elliptically symmetric distribution.

%%%%%%%%%
% SECTION 3
%%%%%%%%%

\section{Shrinkage parameter for elliptical samples}  \label{sec:beta_opt}

Maronna \cite{maronna1976robust} developed M-estimators of scatter matrices originally  within the framework of elliptically symmetric distributions \cite{fang1990symmetric,ollila2012complex}.   The probability density function (p.d.f.) of centered (zero mean) elliptically distributed
random vector, denoted by $\x \sim \mathcal E_\pdim(\bo 0,\M,g)$, is
\beq \label{eq:pdf_ES} 
	f(\x)
	= C_{\pdim,g} |\M|^{-1/2} g\big(\x^\top \M^{-1} \x\big),
\eeq 
where  $\M$ is a positive definite symmetric  matrix parameter, called the scatter matrix, $g: \left[0,\infty\right) \to \left[0,\infty\right)$ is the
\emph{density generator}, which is a fixed function that is independent of
$\x$ and $\M$, and $C_{\pdim,g}$ is a normalizing constant ensuring
that $f(\x)$ integrates to 1. 
The density generator $g$ determines
the elliptical distribution. For example, the MVN distribution $\mathcal N_p(\bo 0,\M)$ is obtained when
$g(t)=\exp(-t/2)$ and the $t$-distribution with $\nu$ d.o.f., denoted $\x \sim t_{\nu}(\bo 0, \M)$,  
 is obtained when $g(t)= (1+  t/\nu)^{-(p+\nu)/2}$. Then the weight function  for the MLE of scatter matrix  corresponds to  the case that the weight function is of the form 
$
u(t) \propto-   g'(t)/g(t).
$
 This yields  \eqref{eq:t_u} as the weight function for 
the MLE of scatter for $t$-distribution.  If the second moments of $\x$ exists, then $g$ can be defined so that $\M$ represents the covariance matrix of $\x$, i.e., $\M=\cov(\x)$; see  \cite{ollila2012complex} for details.

When $\x \sim \mathcal E_\pdim(\bo 0,\M,g)$, then the M-functional $\Mn$ in \eqref{eq:Mfun} is related to underlying scatter matrix parameter $\M$  via the relationship
\beq \label{eq:scale}
\Mn=\sca \M, 
\eeq 
where $\sca>0$ is a solution to an equation
\begin{equation}\label{eq:sca} 
\E \bigg[ \psi \bigg( \frac{ r^2 }{\sca} \bigg)  \bigg] = p,  
\end{equation}
 where $\psi(t)=u(t)t$ and $r= \| \M^{-1/2} \x \|$. 
Often $\sca$  needs to be  solved numerically from \eqref{eq:sca} but in some cases an analytic expression 
can be derived. If  $\x \sim \mathcal E_p(\bo 0,\M,g)$ and the used weight function matches with the data distribution, so 
$u(t) \propto-   g'(t)/g(t)$, then $\sigma = 1$.

Next we derive expressions for  $ \E[\tr(\bo C)^2 ]$ and  $\E[\tr(\bo C^2)]$  appearing in the denominator of $\beta_o^{\mathrm{app}}$ in \eqref{eq:beta0id2}. 
These depend on a constant $\psi_1$ (which depends on weight function $u$ via $\psi(t) = u(t) t$) as follows:
\beq \label{eq:psi1}
\psi_1 = \frac{1}{p(p+2) } \E \bigg[   \psi \!\Big( \frac{r^2 }{\sca}\Big)^2   \bigg]  ,
\eeq
 where the statistical expectation is again computed  w.r.t. distribution of    
 the positive random variable $ r = \| \M^{-1/2} \x \|$. 

\begin{lemma} \label{lem:EtrC} Suppose $\x_1, \ldots, \x_n$  is an i.i.d.\ random sample from  $\mathcal E_p(\bo 0, \M,g)$. Then 
\begin{align*}
\E \big[\tr \! \big(\Mor^2\big) \big]  &= 
\left( 1 + \frac{2 \psi_1-1}{n} \right) \tr(\Mn^2) +  \frac{\psi_1}{n}\tr(\Mn)^2    
\end{align*}
and 
\begin{align*}
\E[\tr(\bo C)^2] 
	&=  \frac{2 \psi_1}{n}   \tr(\Mn^2) +   \Big(1+ \frac{\psi_1-1}{n}\Big) \tr(\Mn)^2
\end{align*}
given that expectation \eqref{eq:psi1} exists. 
\end{lemma}

\begin{proof}
The proof is given in Appendix~\ref{app:lem:EtrC}.
\end{proof}

This then yields the following main result.  

\begin{theorem}\label{th:beta0ell}
Let $\x_{1},\ldots,\x_{n}$ denote an i.i.d. random sample from an
elliptical distribution $\mathcal E_p(\bo 0, \M,g)$.  Then the oracle parameter $\be_o^\mathrm{app}$ that
minimizes the MSE in \autoref{th:beta0}     is
\begin{align} \label{eq:beta0ell}
 \be_o^{\mathrm{app}}    
 &=    \dfrac{ n(\gamma-1) }{ n(\gamma-1)(1- 1/n)  + \psi_1(1-1/p)(2 \gamma+p)}    
\end{align}
 where   $\gamma$ is defined in~\eqref{eq:gamma} and $\psi_1$ in \eqref{eq:psi1}. 
 \end{theorem}

\begin{proof} Follows from \autoref{th:beta0} after substituting the values for $\E \big[\tr \! \big(\Mor^2\big) \big]$ and  $\E\big[\tr(\bo C)^2\big]$  
derived in \autoref{lem:EtrC} into the denominator of $\be_o^{\mathrm{app}}$ in \eqref{eq:beta0id2}. 
\end{proof}

A closely related result is derived in  \cite[Theorem~1]{ashurbekova2020optimal}. Namely,  \cite{ashurbekova2020optimal} considers oracle estimator as in \eqref{eq:Mor2} but using a shrinkage target equal to identity matrix $\I$ instead of $[\tr(\bo C)/p] \I$ as in this paper. This is due to the fact that \cite{ashurbekova2020optimal} assumes that $\tr(\M)=p$. Another difference is that   \cite{ashurbekova2020optimal}  assumes that $\Mn=\M$ (so $\sca=1$), i.e.,  that the used M-estimator is consistent to the scatter matrix of the underlying elliptical population. This  assumption implies knowledge of the underlying elliptical distribution in which  case it is natural to use the ML-weight  $u_{\text{ML}}(t)= -2 g(t)/g'(t)$ as was done in \cite{ashurbekova2020optimal}. 
 Thus \autoref{th:beta0ell} compared to  \cite[Theorem~1]{ashurbekova2020optimal} holds in the more general case when the scale $\tr(\M)$ is not known {\it apriori} and no assumption on the knowledge of the elliptical distribution is imposed.  Furthermore, in the next subsesction, we extend the result to the complex-valued case which was not considered in \cite{ashurbekova2020optimal}.

\autoref{lem:EtrC}  also allows to construct an unbiased estimate of  $\vartheta = \tr(\Mn^2)/p$ as is  shown next. 

\begin{theorem}\label{th:eta2ell} 
Suppose $\x_1, \ldots, \x_n$  is an i.i.d.\ random sample from  $\mathcal E_p(\bo 0, \M,g)$. 
Then an unbiased estimate of $ \vartheta = \tr(\Mn^2)/p$ for any finite $n$ and any $p$ is
 \begin{align} \label{eq:hatvartheta}
  \hat \vartheta & =   b_n \left(\frac{\tr(\bo C^2)}{p} -  \psi_1 a_n \,   \frac{p}{n}  
	  \, \left[ \frac{\tr(\bo C)}{p} \right]^2 \right)   ,
 \end{align}
where 
\begin{align} \label{eq:an_and_bn}
	a_n =  \frac{n}{n+\psi_1 -1}   \ \mbox{ and } \  
	b_n = \frac{n}{n-1} \left( \frac{n-1 + \psi_1 }{n-1 + 3 \psi_1}  \right)
\end{align}
given that expectation \eqref{eq:psi1}, defining $\psi_1$, exists. 
\end{theorem}

\begin{proof} First note that 
\beq \label{eq:expec_vartheta} 
\E[ \hat \vartheta ] = b_n (\E[\tr(\bo C^2)]/p - a_n (p/n) \E[ \tr(\bo C)^2]/p^2).
\eeq 
Then substituting the values of $\E[\tr(\bo C)^2]$ and $\E \big[\tr \! \big(\Mor^2\big) \big]$ from \autoref{lem:EtrC} into  \eqref{eq:expec_vartheta} yields $\E[\hat \vartheta ] =  \tr(\Mn^2)/p$. 
\end{proof} 

It is instructive to consider in more detail the case of Gaussian loss. In this case, $\bo C$ equals the SCM, $\bo C= \S$, and 
the statistic $\hat \vartheta$ no longer  depends on the unknown $\Mn$.  Furthermore, if data is generated from a Gaussian distribution $ \mathcal N_p(\bo 0, \M)$,  then  $\psi_1=1$, $\Mn=\M$   and  $\hat \vartheta $ in \eqref{eq:hatvartheta} reduces to the estimator that is identical to one proposed by \cite[Lemma 2.1]{srivastava2005some} in the case that location is known ($\bom \mu=\bo 0$); see also \cite[Theorem~2]{li2017estimation}. Moreover,   \cite[Theorem~4]{ollila2019optimal} is obtained in the general elliptical case, again  assuming known location parameter ($\bom \mu=\bo 0$).

%-- SUBSECTION 3.1

\subsection{An extension to the complex-valued case} \label{sec:complex_ext}

Consider the case that  $\x_i $, $i=1,\ldots,n$ are complex-valued (i.e., $\x_i \in \C^p$) and represent a random sample from a (circular) complex elliptically symmetric (CES) distribution \cite{ollila2012complex}.   The p.d.f. of a random vector $\x \in \C^p$ with centered (zero mean) 
CES distribution, denoted using same notation $\x \sim \mathcal E_\pdim(\bo 0,\M,g)$, is
\[
	f(\x)
	= C_{\pdim,g} |\M|^{-1} g\big(\x^\hop \M^{-1} \x\big), \qquad \x \in \C^p, 
\]
where  $\M$ is the positive definite  hermitian (PDH) matrix parameter, called the scatter matrix, $g: \left[0,\infty\right) \to \left[0,\infty\right)$ is the
density generator, which is a fixed function that is independent of $\x$ and $\M$, and $C_{\pdim,g}$ is a normalizing constant ensuring
that $f(\x)$ integrates to 1. 

An M-estimator of scatter matrix $\hat \M$ is a PDH matrix that solves an estimating equation
\beq \label{eq:Mest_c}
\hat \M = \frac{1}{n} \sum_{i=1}^n u(\x_i^\hop \hat \M^{-1} \x_i) \x_i \x_i^\hop ,
\eeq 
where $u: [0, \infty) \to [0,\infty)$ is a non-increasing weight function. Again, $u(t)=1$ gives the SCM $\bo S = \frac 1 n \sum_i \x_i \x_i^\hop$, 
Huber's and Tyler's weight functions are  as earlier, stated in \eqref{eq:huber_u}  and 
\eqref{eq:tyl_u}, respectively, whereas $ u_{\mbox{\tiny T}}(t ; \nu) = \frac{2 \pdim  + \nu }{\nu +2 t}$ corresponds to MLE 
of the scatter matrix parameter when sampling from a $p$-variate complex $t$-distribution with  $\nu$ d.o.f.  We refer to
 \cite{ollila2012complex,zoubir2018robust} for more details.  We may now define the shrinkage M-estimator as in \eqref{eq:shrinkMest}. Definitions 
  \eqref{eq:Mfun}-\eqref{eq:Mor2} hold also in the complex-valued case with obvious modifications  (replacing transpose by the Hermitian transpose). 
 
\autoref{th:beta0}  did not require an assumption that random vectors are real-valued, i.e., it  holds also when  $\x_1, \ldots, \x_n$ are i.i.d. complex-valued random vectors. This means that  we only need to derive expectations in \autoref{lem:EtrC} in the case of random sampling from a CES distribution. First,  we define the parameter $\psi_1$ in the complex-valued case as  
\beq \label{eq:psi1_c}
\psi_1 = \frac{1}{p(p+1) } \E \Big[   \psi \!\bigg( \frac{r^2}{\sca}\Big)^2   \bigg]  ,
\eeq
 where the expectation is w.r.t.   $r= \| \M^{-1/2} \x\|$,   where $\x  \sim \mathcal E_p(\bo 0, \M,g)$. The analog of \autoref{lem:EtrC} to complex case is given next. 
 
 \begin{lemma} \label{lem:EtrC_complex} Suppose $\x_1, \ldots, \x_n$  is an i.i.d.\ random sample  from a (circular) complex elliptically symmetric distribution $ \mathcal E_p(\bo 0, \M,g)$. Then 
\begin{align*}
\E \big[\tr \! \big(\Mor^2\big) \big]  
	&= \left( 1 + \frac{\psi_1-1}{n} \right) \tr(\Mn^2) +  \frac{\psi_1}{n}\tr(\Mn)^2   
\end{align*}
and 
\begin{align*}
\E[\tr(\bo C)^2] 
	&=  \frac{\psi_1}{n}   \tr(\Mn^2) +   \Big(1+ \frac{\psi_1-1}{n}\Big) \tr(\Mn)^2
\end{align*}
given that expectation \eqref{eq:psi1_c} exists. 
\end{lemma}

\begin{proof}
The proof is given in Appendix~\ref{app:lem:EtrC_complex}.
\end{proof}

Plugging in the expectations above into $\be_o^{\text{app}}$ derived in  \autoref{th:beta0}  yields the following result.  

\begin{theorem}\label{th:beta0ell_complex}
Let $\x_{1},\ldots,\x_{n}$ denote an i.i.d. random sample from a (circular) $p$-variate complex elliptically symmetric  
distribution $\mathcal E_p(\bo 0, \M,g)$ and assume that expectation \eqref{eq:psi1_c} exists. 
 Then the oracle parameter $\be_o^\mathrm{app}$ that
minimizes the MSE in \autoref{th:beta0}  is
\begin{align} \label{eq:beta0ell_c}
 \be_o^{\mathrm{app}}   
 &=    \dfrac{n(\gamma-1)}{  n(\gamma-1)(1-1/n)  + \psi_1(1-1/p)(\gamma+p)}     , 
\end{align}
 where  $\gamma$ is defined in~\eqref{eq:gamma} and $\psi_1$ in \eqref{eq:psi1_c}. 

Furthermore, $\hat \vartheta$ defined as in \eqref{eq:hatvartheta} with 
\begin{align*}
	a_n &=  \frac{n}{n+\psi_1 -1}    \ \mbox{ and } \  
	b_n &= \frac{n}{n-1} \left( \frac{n-1 + \psi_1 }{n-1 + 2 \psi_1}  \right)
\end{align*}
  is an unbiased estimate of $ \vartheta = \tr(\Mn^2)/p$ for any finite $n$ and any $p$. 
\end{theorem}

%-- SUBSECTION 3.2

\subsection{Computing the shrinkage parameter} \label{subsec:pract_comp}

In order to construct a data-adaptive estimate of  $\beta_o^{\text{app}}$ (either in real- or complex-valued cases), all we need to estimate is the sphericity $\gamma$ and the constant $\psi_1$.   An estimate $\hat \psi_1$ of $\psi_1$ is constructed separately for each  weight function (Gaussian, Huber's, MVT- and Tyler's weight function) in  \autoref{sec:special_cases}. However, it  is also possible to use an empirical (sample mean) version of  \eqref{eq:psi1}. 
Next we discuss computation of the sphericity estimator $\hat \gamma$. 

 As an estimator $\hat \gamma$  we use the estimate derived in \cite{zhang2016automatic}, which was named as Ell1-estimator in \cite{ollila2019optimal},  and defined as 
\begin{align}\label{eq:gammahat}
	\hat \gamma^{\text{Ell1*}} =  \frac{n}{n-1}\left(  p \, \tr\left( \frac{1}{n} \sum_{i=1}^{n}  \frac{\x_{i}  \x_{i}^\top}{\|\x_{i}\|^{2}}  \right)  - \frac{p}{n}\right) 
\end{align}
which for complex-valued case is defined analogously (transpose replaced with the Hermitian transpose). 

Next recall that  $\hat \vartheta$ defined in \eqref{eq:hatvartheta} is an unbiased estimator of $\tr(\Mn^2)/p$. This statistics depends on  $\bo C$ and $\psi_1$  which are unknown, but a plug-in estimate of $\hat \vartheta$ can be constructed by replacing $\bo C$  and $\psi_1$ with $\hat \M$ 
and $\hat \psi_1$, respectively. Dividing this plug-in statistic further by $(\tr(\hat \M)/p)^2$,  leads to another estimator of sphericity, named as Ell2-estimator, and defined as
\begin{align}\label{eq:gammahat_ell2}
	\hat \gamma^{\text{Ell2*}} =  \hat b_n   \left( \frac{p\tr(\hat \M^2)}{\tr(\hat \M)^2} -  \hat \psi_1 \hat a_n  \frac{p}{n}  \right) , 
\end{align}
where the constants $\hat a_n $ and $ \hat b_n$ are as defined in \autoref{th:eta2ell} (and \autoref{th:beta0ell_complex} in complex case) but with $\psi_1$ replaced by its estimate $\hat \psi_1$.  When one uses Gaussian weight function, then $\hat \M=\bo S$ and $\hat \psi_1 = 1+ \hat \kappa$, where $\hat \kappa$ is an estimate of elliptical kurtosis  ({\it cf.} \autoref{subsec:RSCM}). In  this case, $\hat \gamma^{\text{Ell2*}}$  corresponds to Ell2-estimator of sphericity proposed in \cite[Sect.~IV.B]{ollila2019optimal}. 

In order to guarantee that the estimators remain in the valid interval, $1\leq \gamma \leq p$, we use 
\begin{equation}\label{eq:hatgamma_ell1_eq2}
   \hat \gamma^{\textup{Ell1}}  
   = \min(p, \max(1, \hat \gamma^{\textup{Ell1*}}))
\end{equation}
as our final estimator (and similarly for Ell2-estimator). The related shrinkage parameter can then computed as  
\[
\beta= \beta_o^{\mathrm{app}}(\hat \gamma^{\text{Ell1}}, \hat \psi_{1}), 
\]
and again similarly for Ell2-estimator.

%%%%%%%%%
% SECTION 4
%%%%%%%%%

\section{Important special cases} \label{sec:special_cases} 
%-- SUBSECTION 4.1 (RSCM)
\subsection{Regularized SCM (RSCM) estimator}   \label{subsec:RSCM} 

If one uses Gaussian weight function $u_{\text{G}}(t) \equiv 1$, then a necessary assumption is that the underlying elliptical distribution possesses finite 4th-order moments. As discussed earlier, one may then assume w.l.o.g. that the scatter matrix parameter equals the covariance matrix, i.e.,  $\M= \cov(\x)$. %$\Mn=  \M$ and $\sigma=1$. 
When $u_{\text{G}}(t) \equiv 1$, one has that  $\hat \M= \S$ and  $\bo C_\be=\S_\be$ and hence $\be_o  = \be_o^\mathrm{app}$, meaning that the approximate MMSE solution is exact in this case. Finally,  we may relate $\psi_1$ with an elliptical kurtosis~\cite{muirhead:1982} parameter $\kappa$:
\beq \label{eq:kappa} 
\psi_1 = 1+ \ka=  \begin{cases} 
\dfrac{\E [r^4]}{p(p+2)}   , & \mbox{real case} \\%\vspace{2pt} \\ 
 \dfrac{\E [ r^4]}{p(p+1)} ,  & \mbox{complex case}  
 \end{cases} 
 \eeq  
 where the expectation is over the distribution of the random variable $r=  \| \M^{-1/2}\x \|$.  The elliptical kurtosis parameter is defined as a generalization of the kurtosis parameter to the vector case, and as such it vanishes (so $\ka=0$) when $\x$ has MVN distribution (denoted $\x  \sim \mathcal N_p(\bo 0, \M)$).  
 Since $\bo S$ exists for any $n \geq 1$, we can drop the assumption that  $n>p$ in this case. 
 
\begin{corollary}  \label{cor:RSCM} Let $\x_{1},\ldots,\x_{n}$ denote an i.i.d. random sample from an (real or complex)
elliptical distribution $\mathcal E_p(\bo 0, \M,g)$ with finite 4th order moments and covariance matrix $\M=\cov(\x)$. 
Then for the shrinkage SCM estimator $\S_\be$ in \eqref{eq:LW} one has that 
\begin{align} 
\be_o   &=  \arg \min_\be  \, \mathbb{E} \big[ \big\| \S_\be - \M \|^2_{\mathrm{F}} \big] 
=\dfrac{ n(\gamma-1)}{   n(\gamma-1) + p + a} ,  \label{eq:beta0_ell_scm}
\end{align} 
where 
\[
a = 
 \begin{cases} 
 \ka(2 \gamma(1-1/p) + p-1)   + \gamma (1-2/p), &\mbox{real case}  \\ 
 \ka(\gamma(1-1/p) + p-1)   - \gamma/p,&\mbox{complex case}
 \end {cases}  
\] 
\end{corollary}

\begin{proof}  The result follows from \autoref{th:beta0ell} and~\autoref{th:beta0ell_complex} since $\bo C_\be = \bo S_\be$ and the M-functional for Gaussian loss is  $\Mn = \cov(\x) = \M$ and $\sigma=1$.  Since for Gaussian loss, $\psi(t)= t$, we notice from  \eqref{eq:psi1} and  \eqref{eq:psi1_c} hat 
\beq \label{eq:psi1_RSCM}
\psi_1 = 1 + \kappa. 
\eeq
 Plugging $\psi_1=1+\ka$ into  $\beta_o^{\text{app}}$ in \autoref{th:beta0ell} and \autoref{th:beta0ell_complex} yields  the stated expressions, respectively. 
\end{proof}

 The elliptical kurtosis parameter $\kappa$ can be easily estimated using the following relationship to kurtosis even in the cases when $p>n$. 
First, recall that kurtosis of a random variable $x$ in the real and complex case is defined as 
\beq \label{eq:kurt} 
 \mathsf{kurt}(x) = \frac{\E[ x^4]}{ (\E[x^2])^2} - 3 
 \quad \mbox{and} \quad  \mathsf{kurt}(x)  =    \frac{\E[ | x|^4]}{ (\E[ |x|^2])^2} - 2 , 
 \eeq 
 respectively. Kurtosis vanishes when the random variable has real or complex Gaussian distribution with variance $\E[|x|^2]$. 
 The following result establishes the relationship of elliptical kurtosis parameter with marginal kurtosis. 
 
 \begin{lemma}  \label{lem:kappa} 
Assume that  $\x$ is a random vector from real or complex elliptically symmetric distribution with covariance matrix $\M=\cov(\x)$ possessing finite 4th order moments.  
Then 
\beq  \label{eq:kappa_relation} 
 \kappa =  \begin{cases} \frac{1}{3}  \mathsf{kurt}(x_j), & \mbox{real case} \\    \frac{1}{2}  \mathsf{kurt}(x_j), & \mbox{complex case}  \end {cases}  
\eeq 
where $x_j$ is any $j$th component of $\x$ ($j \in \{1,\ldots,p\}$).
\end{lemma} 

\begin{proof}
The proof is given in Appendix~\ref{app:lem:kappa}.
\end{proof}

Since all marginal variables possess the same kurtosis,  an estimate $\hat \kappa$ can be formed simply as the mean of marginal sample kurtosis statistics. 
This is the same estimate of the elliptical kurtosis proposed in  \cite{ollila2019optimal}.
% In the real case, depending on which estimator of sphericity one uses, one obtains RSCM-Ell1 estimator (using \eqref{eq:gammahat}) or RSCM-Ell2 estimator when the sphericity is estimated as a normalized SCM (see \cite{ollila2019optimal} for details). 
Note that \cite{ollila2019optimal} only considered the real-valued case, and thus \autoref{cor:RSCM} allows us to extend the RSCM estimator in \cite{ollila2019optimal} to complex-valued case.

In the sequel, we use acronym {\bf RSCM-Ell1} to refer to estimator $\S_{\be}$ with $\be$ computed 
as $\be=\beta_o(\hat \kappa,\hat \gamma^{\text{Ell1}})$ with $\beta_o$ given by   \eqref{eq:beta0_ell_scm} and $\hat \gamma^{\text{Ell1}}$ being the estimate of sphericity defined in \eqref{eq:hatgamma_ell1_eq2} and  $\hat \kappa$ an estimate of elliptical kurtosis described above.  An {\bf RSCM-Ell2} estimator is  defined similarly but now using Ell2-estimator of sphericity. 

A natural competitor for RSCM-Ell1 or RSCM-Ell2 estimators (at least in the real-valued case) is  the estimator proposed by Ledoit and Wolf \cite{ledoit2004well},  referred to as {\bf LWE}.   We note that  LWE  also uses RSCM $\S_{\be}$, but the parameter $\be$ is computed in a different manner.  An extra benefit of our approach is that an estimator of the optimal shrinkage parameter can be computed for real- or complex-valued observations while LWE assumes real-valued observations.

%-- SUBSECTION 4.2 (RHub)
\subsection{Regularized Huber's M-estimator (RHub)}  \label{subsect:huber} 

Next consider the Huber's weight function $u_{\mbox{\tiny H}} (t;c)$ in  \eqref{eq:huber_u}. Note that $b>0$ is a scaling constant; if  $\hat{\bom \Sigma}$ is Huber M-estimator of scatter when $b = 1$, then the Huber M-estimator of scatter when $b = b_o$ is simply $b_o\hat{\bom \Sigma}$. The scaling constant $b$ is usually chosen so that the resulting scatter estimator is Fisher consistent  for the covariance matrix at MVN distribution, i.e, $\sigma=1$ when $ \x \sim \mathcal N_\pdim(\bo 0, \M)$. In the real case, this holds when 
\[
b = %\begin{cases} 
 F_{\chi^2_{\pdim+2}}(c^2) + c^2(1-F_{\chi^2_{\pdim}}(c^2))/\pdim,  % , & \mbox{ real-valued case}  \\ 
% F_{\chi^2_{2(\pdim+1)}}(2c^2) + c^2(1- F_{\chi^2_{2 \pdim}}(2 c^2))/\pdim  , & \mbox{ complex-valued case}
%\end{cases} 
\] 
where $F_{\chi^2_p}(\cdot)$ denotes the cumulative distribution function (c.d.f.) of chi-squared distribution with $p$ d.o.f. 
Since  $r^2=\|  \M^{-1/2} \x \|^2$ has a  $\chi^2_p$-distribution when $ \x \sim \mathcal N_\pdim(\bo 0, \M)$, the tuning constant $c^2$  is chosen  as $q$th upper quantile of $\chi^2_\pdim$-distribution:
\beq \label{eq:csq} 
q = \Pr( r^2 \leq c^2 )  \Leftrightarrow  F^{-1}_{\chi^2_p}(q)  = c^2 
\eeq  
for some $q \in (0,1]$.  Tuning constant $c$ and scaling factor $b$ can be determined similarly in the complex-valued case; see \cite{ollila2012complex,zoubir2018robust}  for details. 

Let us define a winsorized observation $\bo w$ based on $\x \sim \mathcal E_p(\bo 0, \M,g)$ as 
\[
\bo w  = \wins(\x; \M, c)= \frac{1}{\sqrt{b}} \begin{cases}  \x,  &  \,  \|  \M^{-1/2} \x\|^2  \leqslant  c^2 \\ 
 c \dfrac{\x}{\| \M^{-1/2} \x\|},  & \,   \| \M^{-1/2} \x\|^2 > c^2 \end{cases}  
\]
 where $c$ is the threshold $c$ of Huber's weight function and $b$ is the respective scaling factor. 
The winsorized  r.v. $\bo w$ also has an elliptically symmetric distribution since the contours remain elliptical in shape (so the p.d.f. is still defined by \eqref{eq:pdf_ES} but for a truncated density generator $g$) and thus it shares the properties of elliptical random vectors. 
 
If we take $\sigma=1$ (which holds at least when $\x$ has MVN distribution), then the constant  $\psi_1$ can be written as 
\begin{align*}  
\psi_1   &= \frac{ \E[  \psi_{\mbox{\tiny H}}^2( \|  \M^{-1/2} \x \|^2 ;c )]}{p (p+2)} 
= \frac{\E \big[  \|  \M^{-1/2}  \bo w \|^4 \big]  }{p(p+2)}   \\ 
&= 1 +  \kappa_{\bo w}
\end{align*} 
where  $\kappa_{\bo w}$ is the elliptical kurtosis parameter ({\it cf.}  \autoref{lem:kappa}) of an elliptical random vector $\bo w$.  
An estimate $\hat \psi_1$  of $\psi_1$ can be then calculated similarly  as $\psi_1$ for RSCM-Ell1 or RSCM-Ell2 estimators defined earlier (recall relation 
\eqref{eq:psi1_RSCM}).  The only difference is that $\kappa$ is now computed for winsorized data $\{\bo w_i \}_{i=1}^n$, where $\bo w_i=\wins(\x_i;  \hat \M, c)$ and $\hat \M$ denotes the Huber's M-estimator. 

In the sequel, we use acronym {\bf RHub-Ell1}   or {\bf RHub-Ell2} to refer to shrinkage M-estimator $\hat \M_\be$ that uses Huber's weight $u(\cdot) = u_{\mbox{\tiny H}}(\cdot; c) $ with  threshold  $c^2$ determined from \eqref{eq:csq} for user specified $q$ and shrinkage parameter  $\be = \beta_o^{\mathrm{app}}(\hat \gamma^{\text{Ell1}}, \hat \psi_{1})$. 
or $\be = \beta_o^{\mathrm{app}}(\hat \gamma^{\text{Ell2}}, \hat \psi_{1})$, respectively. 
 
%-- SUBSECTION 4.3 (RTyl)
\subsection{Regularized Tyler's M-estimator (RTyl)}   \label{subsec:RTyl} 

Let $\bo V$ denote a \emph{shape matrix} (normalized scatter matrix), defined as 
\[
\bo V =  p \M/\tr(\M) ,
\]
where $\M$ denotes the scatter matrix parameter of the ES distribution. Note that $\tr(\bo V)=p$.  If one uses Tyler's weight function in  \eqref{eq:tyl_u}, then  \eqref{eq:scale} holds with  $\sigma=p/\tr(\M)$, i.e., $\Mn = \bo V$, that is, 
Tyler's M-estimator is an estimate of  the shape matrix. The following result hence follows at once from \autoref{th:beta0ell} and \autoref{th:beta0ell_complex}.

\begin{corollary}  \label{cor:beta_tyl} Let $\x_{1},\ldots,\x_{n}$ denote an i.i.d. random sample from a (real or complex) 
elliptical distribution $\mathcal E_p(\bo 0, \M,g)$. When using Tyler's weight \eqref{eq:tyl_u}, it holds that 
\begin{align*} 
\be_o^{\mathrm{app}} &=  \arg \min_\be  \, \mathbb{E} \big[ \big\| \bo C_\be- \bo V \|^2_{\mathrm{F}} \big]  \\
%&=  \dfrac{ \gamma-1}{  (\gamma  -1)(1-1/n)  +   [(p-1)/(p-2)](2 \gamma+p)/n} 
&= \begin{cases}  \dfrac{ n(\gamma-1) }{  n(\gamma-1)(1-\frac{1}{n})  +   \frac{p-1}{p+2}(2 \gamma+p)} , &\mbox{real case}  \vspace{0.1cm} \\ 
\dfrac{ n(\gamma-1)}{  n(\gamma-1)(1-\frac{1}{n})  +   \frac{p-1}{p+1}(\gamma+p)} ,  &\mbox{complex case}\end{cases} 
\end{align*} 
\end{corollary}

\begin{proof}  The real-valued case  follows by noting that  $\psi_1$ in \eqref{eq:psi1}  is equal to 
$\psi_1 = p/(p+2)$ for all random vectors $ \x \sim \mathcal E_p(\bo 0, \M,g)$ regardless of $g$ (i.e., the functional form of the density generator) and  that $\Mn = \bo V$.  Plugging  $\psi_1=p/(p+2)$ into \eqref{eq:beta0ell}  yields  the stated expression. The complex-valued case follows similarly. 
\end{proof} 

Since Tyler's M-estimator verifies $\tr(\hat \M)=p$, the shrinkage estimator in  \eqref{eq:shrinkMest} simplifies to 
\beq \label{eq:shrink_tyl}
\hat \M_\be = \be \hat \M + (1-\be) \bo I, 
\eeq 
where $\hat \M$ is Tyler's M-estimator, i.e., an M-estimator based on weight \eqref{eq:tyl_u}. 
By {\bf RTyl-Ell1} we refer to  \eqref{eq:shrink_tyl},  where the shrinkage parameter is computed as $\be = \beta_o^{\mathrm{app}}(\hat \gamma^{\text{Ell1}})$ with $\be_o^\mathrm{app}$ given by  \autoref{cor:beta_tyl} and  $\hat \gamma^{\text{Ell1}}$ by \eqref{eq:hatgamma_ell1_eq2}.

A related regularized Tyler's estimator was proposed  by  \cite{abramovich2007diagonally} as the limit of the algorithm
\beq \label{eq:CWH} 
\begin{aligned}
\M_{k+1}     &\leftarrow  \be \frac{\pdim}{\ndim} \sum_{i=1}^\ndim \frac{\x_i \x_i^\top}{\x_i^\top \bo V_{k}^{-1} \x_i} + (1-\be) \I \\ 
\bo V_{k+1} &\leftarrow  \pdim \M_{k+1}/\tr(\M_{k+1}),
\end{aligned}
\eeq 
where $\be \in (0,1)$ is a fixed shrinkage parameter. This algorithm represents a diagonally loaded  version of the fixed-point algorithm  given for Tyler's M-estimator.  Uniqueness and convergence of  the recursive algorithm has been later derived in \cite{chen2011robust,sun2014regularized}. 
By {\bf CWH} estimator we now refer to estimator obtained by iterating \eqref{eq:CWH} using same value $\be = \beta_o^{\mathrm{app}}(\hat \gamma^{\text{Ell1}})$ as for RTyl-Ell1.   An interesting question then is how different  is RTyl-Ell1 in its performance from CWH. We explore this by simulation studies later.   This is interesting as the former is simply shrinking the eigenvalues of Tyler's M-estimator towards its grand mean where as the latter does not have an explicit connection to Tyler's M-estimator $\hat \M$ for any $\be \in (0,1)$.  
%-- SUBSECTION 4.4 (RMVT)

\subsection{Regularized M-estimator for MVT distribution (RMVT)} \label{subsect:tMLE}

We assume that the data is arising from a MVT distribution $t_\nu(\bo 0,\M)$ but the d.o.f. parameter $\nu$ is unknown and is adaptively estimated  from the data using \autoref{alg:nu} explained below.  Once $\hat \nu$ is found, we use function $u(\cdot) = u_{\mbox{\tiny T}}(\cdot ; \hat \nu)$ to compute the underlying M-estimator $\hat \M$ for the postulated MVT distribution. 

Yet, we need to address the question of how the constant $\psi_1$ is computed.  Due to data adaptive estimation of $\nu$, we can assume that  $\sigma \approx1$  since the scaling factor $\sigma$ equals unity for an MLE of the scatter matrix parameter.
  We use the fact that for MVT distribution (i.e., when $\x \sim t_\nu(\bo 0, \M)$),  the $\psi_1$ parameter is\footnote{Note that the $t$-weight in the complex case is  \cite{ollila2012complex}: $ u_{\mbox{\tiny T}}(t ; \nu) = \frac{2\pdim  + \nu }{\nu +2 t}$. }  
\[
\psi_1   =  \begin{cases} \dfrac{p+\nu}{2+p+\nu}   , & \mbox{real case} \vspace{2pt} \\
 \dfrac{2p+\nu}{2+2p+\nu}   , & \mbox{complex case} 
\end{cases} . 
\]
Hence, a natural estimate is $\hat \psi_1 = (p + \hat \nu)/(2+p + \hat \nu)$ in the real case.  An estimate of $\hat \psi_1$ is constructed similarly in the complex case. We use acronym {\bf RMVT-Ell1} to refer to shrinkage M-estimator $\hat \M_\be$ that uses $u(\cdot; \hat \nu) $ with    shrinkage parameter  calculated by  $\be = \beta_o^{\mathrm{app}}(\hat \gamma^{\text{Ell1}}, \hat \psi_{1})$.  {\bf RMVT-Ell2} is constructed similarly, but now Ell2-estimator of sphericity $\hat \gamma^{\text{Ell2}}$ is used.

Next we discuss our approach for estimating $\nu$ from the data. 
Assume  $\x \sim t_{\nu}(\bo 0, \M,g)$ and denote $\eta= \tr(\M)/p$. Then, 
\[
\bo R = \cov(\x)= (\nu/(\nu-2)) \M
\]
and hence $\tr( \bo R)/p= (\nu/(\nu-2)) \tr(\M)/p$.  This means that 
\[
\frac{\nu}{\nu -2} = \frac{ \tr(\bo R)}{\tr(\M)}  = \eta_{\text{ratio}} 
\]
from which we obtain the relation
\beq \label{eq:eta_ratio}
\nu = \frac{2 \eta_{\text{ratio}}}{\eta_{\text{ratio}}-1}.
\eeq 
The above relation holds true in both real and complex cases.  Then given an estimate $\hat \M$ of $\M$,  we may compute an estimate $ \hat \eta_{\text{ratio}}= \tr(\S)/\tr(\hat \M)$ which in turn provides an  estimate $\hat \nu$ via \eqref{eq:eta_ratio}. 
 This  gives rise to an iterative algorithm  to estimate $\nu$ detailed in  \autoref{alg:nu}. In the simulations, the algorithm converged, but already 2 iterations are sufficient to yield accuracy to first decimal; see \autoref{fig:nu} for an illustration.  
The initial estimate  is $\nu_{o}=2/(\max(0,\hat \kappa) +\delta)+ 4$, where $\hat \kappa$ is an estimate of marginal kurtosis explained in  \autoref{subsec:RSCM} (see also \cite{ollila2019optimal} for more details) and $\delta>0$ is a small number. The initial start $\nu_{o}$ is based on the following relationship with elliptical kurtosis parameter,  $\kappa =  2/(\nu-4)$, i.e., $\nu = 2/\kappa + 4$ which holds  true both in real and complex cases. Again the estimate $\hat \nu$ in the complex-valued case is constructed similarly.   Note that also other estimators of $\nu$ are proposed in the literature, for example in \cite{ashurbekova2020optimal}.

\begin{figure}[!t]
\setlength\fwidth{0.57\textwidth}
\centering
\begin{tikzpicture}
\footnotesize
\begin{axis}[%
width=0.56\fwidth,
height=0.44\fwidth,
scale only axis,
xmin=60,
xmax=320,
xlabel style={font=\color{white!15!black}},
xlabel={\small $n$},
 axis x line*=bottom,
 axis y line*=left,
ymin=5,
ymax=8,
ylabel style={font=\color{white!15!black}},
ylabel={$\nu$},
axis background/.style={fill=white},
xmajorgrids,
ymajorgrids,
legend style={legend columns=1,  at={(1.41,1.02)}, legend cell align=left, align=left, fill=none, draw=none} %anchor=south east, 
]
\addplot [color=black,  line width=0.7pt, mark size=2.0pt, mark=x, mark options={solid, black}]
  table[row sep=crcr]{%
60	6.7936080506437\\
80	5.91380723619394\\
100	5.65101166568368\\
120	5.48975031009831\\
140	5.42422033666137\\
160	5.37960082721805\\
180	5.33063498025775\\
200	5.28696491315068\\
240	5.24946288444055\\
320	5.2092635371655\\
};
\addlegendentry{$\hat \nu$  ($T_{max}=15$)}

\addplot [color=blue, line width=0.7pt, mark size=2.0pt, mark=o, mark options={solid, blue}] %dashdotted,  
  table[row sep=crcr]{%
60	6.81728504910592\\
80	5.94515821456576\\
100	5.67158376717789\\
120	5.50506291607559\\
140	5.43711743577757\\
160	5.39071843292585\\
180	5.3406128593158\\
200	5.29588479566354\\
240	5.25684078650721\\
320	5.21478356032472\\
};
\addlegendentry{$\hat \nu$ ($T_{max}=2$)}

\addplot [color=red, line width=0.7pt, mark size=2.0pt, mark=square, mark options={solid, red}]% dashed, 
  table[row sep=crcr]{%
60	7.64294673526727\\
80	7.2369818368713\\
100	6.94849846776953\\
120	6.70981868782695\\
140	6.61546682196353\\
160	6.521927790159\\
180	6.443460009184\\
200	6.35847863225821\\
240	6.2390777269822\\
320	6.09303030863221\\
};
\addlegendentry{$\nu_0$}

\end{axis}

\begin{axis}[%
width=0.21\fwidth,
height=0.19\fwidth,
at={(0.3\fwidth,0.24\fwidth)},
scale only axis,
xmin=99,
xmax=101,
xtick={ 99, 100, 101},
ymin=5.64,
ymax=5.69,
 yticklabel style={/pgf/number format/.cd,fixed,precision=3},
axis background/.style={fill=white},
legend style={legend cell align=left, align=left, draw=white!15!black}
]
\addplot [color=black,line width=0.9pt, mark size=2.0pt, mark=x, mark options={solid, black}]
  table[row sep=crcr]{%
80	5.91380723619394\\
100	5.65101166568368\\
120	5.48975031009831\\
};

\addplot [color=blue, dashdotted, line width=0.9pt, mark size=2.0pt, mark=o, mark options={solid, blue}]
  table[row sep=crcr]{%
80	5.94515821456576\\
100	5.67158376717789\\
120	5.50506291607559\\
};

\addplot [color=red, dashed, line width=0.9pt, mark size=2.0pt, mark=square, mark options={solid, red}]
  table[row sep=crcr]{%
80	7.2369818368713\\
100	6.94849846776953\\
120	6.70981868782695\\
};

\end{axis}
\end{tikzpicture}% 
\caption{Average $\hat \nu$ by running the \autoref{alg:nu} with different choices of $T_{max}$. Also shown is the initial estimate $\nu_0$. 
The samples are generated from a $p$-variate  $t$-distribution with $\nu=5$ d.o.f., where $\M$ follows the same AR(1) covariance matrix structure 
explained in the simulation set-up of  \autoref{sec:simul}; $\varrho=0.6$ and $p=40$,    As can be noted, $\hat \nu$ converge to  $\nu=5$ as $n$ increases, albeit the convergence is a bit slow.} \label{fig:nu}
\end{figure}
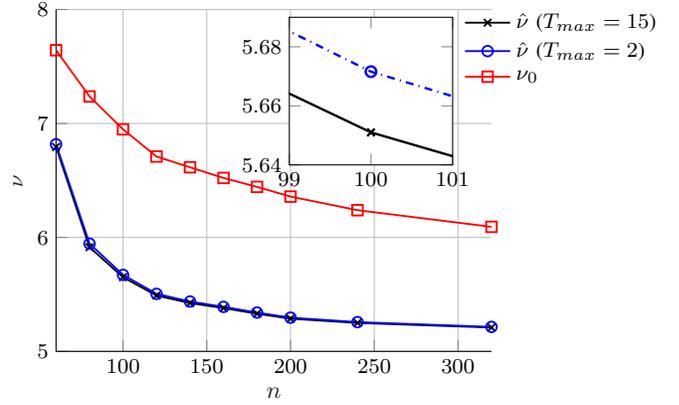

\begin{algorithm}[t]
	\caption{Automatic data-adaptive computation of the d.o.f. parameter $\nu$} \label{alg:nu}
	\DontPrintSemicolon
	\SetKwInOut{Input}{Input}\SetKwInOut{Output}{Output}
	\SetKwFunction{Support}{InitSupport}
	\SetKwInOut{Init}{Initialize}
	\SetAlgoNlRelativeSize{-1}	
	\SetNlSkip{1em}
	\SetInd{0.5em}{0.5em}
	\Input{data matrix $\X$ of size $n \times p$, maximum number of iterations $T_{max}$}			
	\Init{ 
	 Compute $\nu_0 = 2/\max(0,\hat \kappa) + 4$, 
	where $\hat \kappa$ is an estimate of $\kappa$ explained in the text.}			
	
\For{$ \mathsf{t}=0,1, \ldots, T_{max}$}{ 
	
		 Set  $\hat \M_{\mathsf{t}} = \hat \M$, where $\hat \M$ denotes the $t$-MLE based on current estimate  of d.o.f. parameter $\nu = \nu_{\mathsf{t}}$, 	 i.e., solving  the M-estimating equation 
		 \[
		 \hat \M = \frac{1}{n} \sum_{i=1}^n u(\x_i^\top \hat \M^{-1} \x_i) \x_i \x_i^\top
		 \] 
		 with  $ u(\cdot) = u_{\mbox{\tiny T}}(\cdot; \nu_{\mathsf{t}})$ is the MVT-weight function. 
		 
		 Update the ratio $ \hat \eta_{\mathsf{t}}  = \dfrac{\tr(\bo S)}{\tr(\hat \M_{\mathsf{t}})}$. 	
		
		 Upate the d.o.f. parameter  $\nu_{\mathsf{t+1}}  = \dfrac{2 \eta_{\mathsf{t}}}{\eta_{\mathsf{t}}-1}$.
		 
		 \If{ $ | \nu_{\mathsf{t+1}}  -  \nu_{\mathsf{t}} |/ \nu_{\mathsf{t}} < 0.01$ }{break \; }
								 		
	}
	\Output{$\hat \nu = \nu_{\mathsf{t+1}}$} 
\end{algorithm}

%%%%%%%%%
% SECTION 5
%%%%%%%%%

\section{Simulation studies}  \label{sec:simul}

In the simulation study, we generate samples from real ES distributions with a scatter  matrix  $\M$ following an AR(1) structure, $(\M)_{ij} =  \tau \varrho^{|i-j|}$, where $\varrho \in (0,1)$ and  scale  parameter  $\tau = \tr(\M)/\pdim =10$.  When $\varrho
\downarrow 0$, then $\M$ is close to an identity matrix scaled by $\tau$, and when $\varrho
\uparrow 1$, $\M$ tends to a singular matrix of rank 1. 
The results are reported for the proposed shrinkage M-estimators using shrinkage parameter estimates based on  Ell1-estimator of sphericity. However,  for notational convenience,  we drop the suffix -Ell1 from the proposed estimators. Thus the proposed estimators, described in \autoref{sec:special_cases} are referred to as  RSCM, RMVT, RHub, and RTyl.  Furthermore, acronyms LW, CWH and RBLW are used to refer to estimators proposed in  \cite{ledoit2004well} (see also \autoref{subsec:RSCM}), 
\cite{chen2011robust} (see also \autoref{subsec:RTyl} and \eqref{eq:CWH}) and  \cite{chen2010shrinkage}, respectively. {\bf RBLW} is the Rao-Blackwellized LW estimator, but unlike LW estimator, it assumes that the data distribution is Gaussian.  

We also compare to RSCM estimator $\bo S_\be$ in \eqref{eq:LW}, but now the shrinkage parameter $\beta$ is chosen via $k$-fold cross-validation  (CV), where as cross-validation fit  we use $\| \bo S_{\be,\text{tr}} -  \bo S_{\text{val}} \|_{\Fr}$,  where $\bo S_{\text{val}}$ is the SCM based on the validation set (data fold that was left out) and $\bo S_{\be,\text{tr}}$ is the RSCM computed on the training set (data based on remaining folds) for a given $\be$. 
 As a grid for $\be$ for the CV method we use a uniform grid in $[0,1]$ with $0.05$ increments and $5$-fold cross-validation. We call this method as {\bf RSCM-CV} or simply {\bf CV}. All simulation results in this section are averaged over 2000 Monte-Carlo trials. Since $n>p$ is assumed for all estimators expect for RSCM, we do not consider the case low sample regime, $n \leq p$, in our simulation studies. Furthermore, we adopt the MSE (squared Frobenius norm) as our performance metric as it  is used in deriving the optimal shrinkage parameters in this paper. It is important to keep in mind, however, that in the low sample regime and for different performance metrics, the performance differences between estimators can often be noticeable and even quite different than in the $n>p$ regime that is considered here; See e.g., \cite{chen2010shrinkage,coluccia2015regularized,ollila2019optimal} for numerical illustrations and \cite{smith2005covariance,breloy2019intrinsic} for different distances between covariances that could be used instead of the MSE metric.

\subsection{Gaussian data}  

\begin{figure}[!t]
\setlength\fwidth{0.57\textwidth}
\centering
\definecolor{mycolor1}{rgb}{0.14200,0.67200,0.30300}%
\definecolor{mycolor2}{rgb}{0.54200,0.27200,0.60300}%
\begin{tikzpicture}
\footnotesize
\begin{axis}[%
width=0.56\fwidth,
height=0.44\fwidth,
scale only axis,
xmin=60,
xmax=280,
xlabel style={font=\color{white!15!black}},
xlabel={$n$},
ymin=0.0636940145535204,
ymax=0.210697185427197,
ylabel style={font=\color{white!15!black}},
ylabel={\small $\| \hat \M- \M \|_{\mathrm{F}}^2/\| \M \|_{\mathrm{F}}^2$},
axis background/.style={fill=white},
xmajorgrids,
ymajorgrids,
legend style={legend columns=1,  at={(1.3,1.02)}, legend cell align=left, align=left, fill=none, draw=none} %anchor=south east, 
]

\addplot [color=red,  line width=0.7pt, mark size=2.0pt,  mark=triangle, mark options={solid, rotate=180, red}]%dashed, 
  table[row sep=crcr]{%
60	0.207725776439129\\
80	0.172499093853904\\
100	0.147918479258892\\
120	0.128509430218218\\
140	0.114191081731507\\
160	0.102227831828467\\
180	0.0933925963346766\\
200	0.0853894498590014\\
220	0.07858322489493\\
240	0.0729517714716043\\
260	0.0678604297808869\\
280	0.0638871916987028\\
};
\addlegendentry{RMVT}

\addplot [color=blue,  line width=0.7pt, mark size=2.0pt, , mark=o, mark options={solid, blue}]
  table[row sep=crcr]{%
60	0.208505050168925\\
80	0.173359330480369\\
100	0.148784739876927\\
120	0.12926780159694\\
140	0.11493056177474\\
160	0.102872883767607\\
180	0.0940396742021025\\
200	0.0859703978851588\\
220	0.0791624793784054\\
240	0.0734912923970398\\
260	0.0683676614262756\\
280	0.0643588105785364\\
};
\addlegendentry{RHub}

\addplot [color=black, line width=0.7pt, mark size=2.2pt, mark=x, mark options={solid, black}]
  table[row sep=crcr]{%
60	0.207876969030952\\
80	0.172564833510881\\
100	0.147942398075557\\
120	0.128498769237209\\
140	0.114171035430119\\
160	0.102208308553464\\
180	0.09334466875506\\
200	0.0853508252168992\\
220	0.0785374986395278\\
240	0.072911200817477\\
260	0.0678120381336769\\
280	0.0638484943097481\\
};
\addlegendentry{RSCM}

\addplot [color=mycolor1,  line width=0.7pt, mark size=2.0pt, mark=square, mark options={solid, mycolor1}]%dotted,
  table[row sep=crcr]{%
60	0.208901166460334\\
80	0.173145104162067\\
100	0.148275163484884\\
120	0.128855400954153\\
140	0.114401209118349\\
160	0.102385082337128\\
180	0.0934845308622866\\
200	0.0854077001571741\\
220	0.078615546768835\\
240	0.0729629780778958\\
260	0.0678833752418767\\
280	0.0639114632476581\\
};
\addlegendentry{LW}

\addplot [color=mycolor2, line width=0.7pt, mark size=2.2pt, mark=triangle, mark options={solid, mycolor2}]%dotted, 
  table[row sep=crcr]{%
60	0.206983465636785\\
80	0.171741319437712\\
100	0.147245933587506\\
120	0.128037163740915\\
140	0.113709929761856\\
160	0.101878999300753\\
180	0.0930649930088065\\
200	0.0850392934114407\\
220	0.0783011899007216\\
240	0.0727174598446281\\
260	0.0676620147933561\\
280	0.0636940145535204\\
};
\addlegendentry{RBLW}

\addplot [color=gray, line width=0.7pt, mark size=2.2pt, mark=asterisk, mark options={solid, gray}]%dashed, 
  table[row sep=crcr]{%
60	0.210697185427197\\
80	0.174716243012767\\
100	0.149658548978391\\
120	0.129913797947737\\
140	0.115346294002788\\
160	0.103161772855069\\
180	0.0941813573536172\\
200	0.0861046993625721\\
220	0.07917884910182\\
240	0.0734612544577657\\
260	0.0683082194552209\\
280	0.0642673935376015\\
};
\addlegendentry{CV}

\end{axis}

\begin{axis}[%
width=0.25\fwidth,
height=0.23\fwidth,
at={(0.27\fwidth,0.18\fwidth)},
scale only axis,
xmin=119,
xmax=121,
xtick={119, 120, 121},
ymin=0.1273,
ymax=0.1305,
ytick={0.128, 0.129,  0.13},
 yticklabel style={/pgf/number format/.cd,fixed,precision=3},
axis background/.style={fill=white},
legend style={legend cell align=left, align=left, draw=white!15!black}
]
\addplot [color=red,  line width=0.9pt, mark size=2.5pt, mark=triangle, mark options={solid, rotate=180, red}]
  table[row sep=crcr]{%
100	0.147918479258892\\
120	0.128509430218218\\
140	0.114191081731507\\
};

\addplot [color=blue, line width=0.7pt, mark size=2.3pt,  mark=o, mark options={solid, blue}]
  table[row sep=crcr]{%
100	0.148784739876927\\
120	0.12926780159694\\
140	0.11493056177474\\
};

\addplot [color=black,  line width=0.7pt, mark size=2.3pt, mark=x, mark options={solid, black}]
  table[row sep=crcr]{%
100	0.147942398075557\\
120	0.128498769237209\\
140	0.114171035430119\\
};

\addplot [color=mycolor1, line width=0.7pt, mark size=2.0pt, mark=square, mark options={solid, mycolor1}]
  table[row sep=crcr]{%
100	0.148275163484884\\
120	0.128855400954153\\
140	0.114401209118349\\
};

\addplot [color=mycolor2,  line width=0.7pt, mark size=2.3pt, mark=triangle, mark options={solid, mycolor2}]
  table[row sep=crcr]{%
100	0.147245933587506\\
120	0.128037163740915\\
140	0.113709929761856\\
};

\addplot [color=gray, line width=0.7pt, mark size=2.5pt, mark=asterisk, mark options={solid, gray}]
  table[row sep=crcr]{%
100	0.149658548978391\\
120	0.129913797947737\\
140	0.115346294002788\\
};

\end{axis}

\end{tikzpicture}%
\vspace{-0.4cm}
\caption{NMSE as a function of $n$ when samples are drawn from a MVN distribution $\mathcal N_p(\bo 0, \M)$ with an AR(1)  structure; $\varrho=0.6$ and $p=40$.} \label{fig:AR1}
\end{figure}
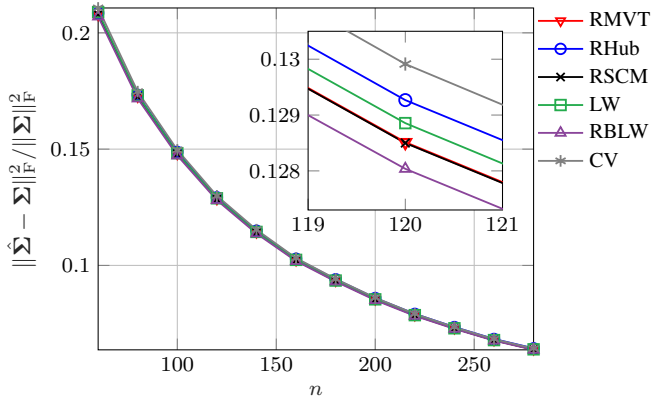

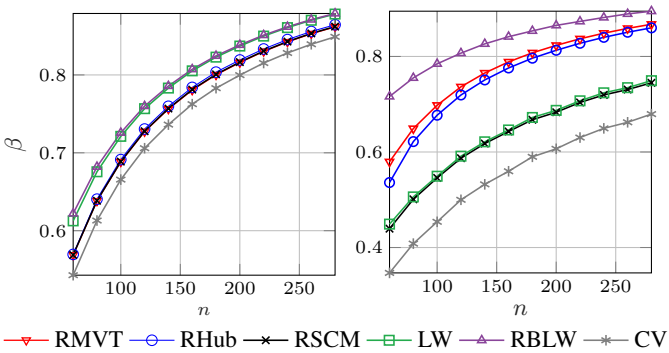
\begin{figure}[!t]
\setlength\fwidth{0.40\textwidth}
\definecolor{mycolor1}{rgb}{0.14200,0.67200,0.30300}%
\definecolor{mycolor2}{rgb}{0.54200,0.27200,0.60300}%
\begin{tikzpicture}
\scriptsize
\begin{axis}[%
width=0.48\fwidth,
height=0.48\fwidth,
at={(0\fwidth,0\fwidth)},
scale only axis,
xmin=60,
xmax=280,
xlabel style={font=\color{white!15!black}},
xlabel={$n$},
ymin=0.542874999999999,
ymax=0.878893089132901,
ylabel style={font=\color{white!15!black}},
ylabel={{\small $\beta$}},
axis background/.style={fill=white},
xmajorgrids,
ymajorgrids,
legend style={legend columns=-1, font=\small, anchor=north west, legend cell align=right, align=right, fill=none, draw=none, overlay, at={(-0.3,-.17)}}	
]
\addplot [color=red,  line width=0.5pt, mark size=1.8pt,  mark=triangle, mark options={solid, rotate=180, red}]
  table[row sep=crcr]{%
60	0.568753506192766\\
80	0.638193649091681\\
100	0.688394036943731\\
120	0.727224379095095\\
140	0.756216269664869\\
160	0.780536828020897\\
180	0.800048599310738\\
200	0.815857836310624\\
220	0.830102172736448\\
240	0.842320875597092\\
260	0.852833778049222\\
280	0.861643552779302\\
};
\addlegendentry{RMVT}

\addplot [color=blue, mark=o, mark options={solid, blue}]
  table[row sep=crcr]{%
60	0.569553959326638\\
80	0.640489581718601\\
100	0.691294178203833\\
120	0.730534355572605\\
140	0.759712374325954\\
160	0.783970233705459\\
180	0.803575336960272\\
200	0.819177206348133\\
220	0.833291876918929\\
240	0.845379340412892\\
260	0.855756233611784\\
280	0.864434359921432\\
};
\addlegendentry{RHub}

\addplot [color=black, line width=0.6pt, mark size=2.0pt,  mark=x, mark options={solid, black}]
  table[row sep=crcr]{%
60	0.568809593325086\\
80	0.638451263573277\\
100	0.688729264679969\\
120	0.727763714581374\\
140	0.756840386115957\\
160	0.781206432562271\\
180	0.800772935646986\\
200	0.816468000366124\\
220	0.830648198487275\\
240	0.842849883092528\\
260	0.853326885159567\\
280	0.862096572211429\\
};
\addlegendentry{RSCM}

\addplot [color=mycolor1, line width=0.6pt, mark size=1.8pt,  mark=square, mark options={solid, mycolor1}]
  table[row sep=crcr]{%
60	0.612599952736911\\
80	0.675654585804067\\
100	0.721132119800316\\
120	0.756680003722132\\
140	0.783116782441253\\
160	0.805238668579996\\
180	0.82283529720188\\
200	0.837078476958378\\
220	0.849859589524714\\
240	0.860744239211582\\
260	0.870187087860488\\
280	0.87814316231541\\
};
\addlegendentry{LW}

\addplot [color=mycolor2, line width=0.6pt, mark size=1.8pt,mark=triangle, mark options={solid, mycolor2}]
  table[row sep=crcr]{%
60	0.62207019548342\\
80	0.682056233830439\\
100	0.725732657579359\\
120	0.759969510850843\\
140	0.785719418093473\\
160	0.807249967019282\\
180	0.824488457019465\\
200	0.838525626431785\\
220	0.851085929463499\\
240	0.861746134575976\\
260	0.8710239713687\\
280	0.878893089132901\\
};
\addlegendentry{RBLW}

\addplot [color=gray, line width=0.6pt, mark size=2.0pt, mark=asterisk, mark options={solid, gray}]
  table[row sep=crcr]{%
60	0.542874999999999\\
80	0.613050000000003\\
100	0.665599999999997\\
120	0.705900000000002\\
140	0.736224999999997\\
160	0.762475000000005\\
180	0.783000000000002\\
200	0.799749999999992\\
220	0.815275000000004\\
240	0.82817500000001\\
260	0.839374999999996\\
280	0.848924999999983\\
};
\addlegendentry{CV}

\end{axis}
\end{tikzpicture}% 
\definecolor{mycolor1}{rgb}{0.14200,0.67200,0.30300}%
\definecolor{mycolor2}{rgb}{0.54200,0.27200,0.60300}%
\begin{tikzpicture}
\scriptsize
\begin{axis}[%
width=0.48\fwidth,
height=0.48\fwidth,
scale only axis,
xmin=60,
xmax=280,
xlabel style={font=\color{white!15!black}},
xlabel={{\small $n$}},
ymin=0.34685,
ymax=0.894759044941558,
ylabel style={font=\color{white!15!black}},
axis background/.style={fill=white},
xmajorgrids,
ymajorgrids,
legend style={legend columns=1, legend cell align=left, align=left, fill=none, draw=none}
]
\addplot [color=red, line width=0.6pt, mark size=1.8pt,  mark=triangle, mark options={solid, rotate=180, red}]% dashed,
  table[row sep=crcr]{%
60	0.579568646917573\\
80	0.649082236801492\\
100	0.697848690131362\\
120	0.73647474993256\\
140	0.765368734264826\\
160	0.788613765175069\\
180	0.807663308712089\\
200	0.823199963357187\\
220	0.837004794133762\\
240	0.848711913206341\\
260	0.858750862655772\\
280	0.867391904429414\\
};
%\addlegendentry{RMVT}

\addplot [color=blue, line width=0.6pt, mark size=1.8pt,   mark=o, mark options={solid, blue}] %dashdotted, 
  table[row sep=crcr]{%
60	0.536054105795048\\
80	0.621780739892285\\
100	0.677017793528086\\
120	0.719327641130192\\
140	0.750852063607823\\
160	0.775901398175192\\
180	0.796381564293534\\
200	0.81292219727506\\
220	0.827585199402952\\
240	0.840029175709647\\
260	0.850699160533814\\
280	0.859878118048861\\
};
%\addlegendentry{RHub}

\addplot [color=black, line width=0.6pt, mark size=2.0pt,    mark=x, mark options={solid, black}]
  table[row sep=crcr]{%
60	0.439453187997791\\
80	0.500979057457806\\
100	0.545772246674932\\
120	0.586982404642373\\
140	0.617653014219047\\
160	0.642921953927016\\
180	0.667481512110169\\
200	0.683227495615254\\
220	0.703713512371713\\
240	0.719407310961488\\
260	0.730687846031397\\
280	0.744811913044292\\
};
%\addlegendentry{RSCM}

\addplot [color=mycolor1, line width=0.6pt, mark size=1.8pt,  mark=square, mark options={solid, mycolor1}] %dotted, 
  table[row sep=crcr]{%
60	0.449011378930349\\
80	0.505972594943444\\
100	0.549284619328713\\
120	0.591273587712397\\
140	0.621313790346825\\
160	0.646096168822703\\
180	0.672294730015853\\
200	0.687145138976202\\
220	0.707170847544069\\
240	0.723935781946755\\
260	0.73469389883358\\
280	0.749313048889081\\
};
%\addlegendentry{LW}

\addplot [color=mycolor2,  line width=0.6pt, mark size=1.8pt, mark=triangle, mark options={solid, mycolor2}] %dotted,
  table[row sep=crcr]{%
60	0.716171793681385\\
80	0.755145419469561\\
100	0.784568577112001\\
120	0.806931812338108\\
140	0.826346020577112\\
160	0.841344004302816\\
180	0.853522369049999\\
200	0.865048550839738\\
220	0.873358515976304\\
240	0.881548241059268\\
260	0.889054460737973\\
280	0.894759044941558\\
};
%\addlegendentry{RBLW}

\addplot [color=gray, line width=0.6pt, mark size=2.0pt, mark=asterisk, mark options={solid, gray}] %dashed, 
  table[row sep=crcr]{%
60	0.34685\\
80	0.408175\\
100	0.453824999999999\\
120	0.4998875\\
140	0.532262500000001\\
160	0.559662499999999\\
180	0.589862499999999\\
200	0.60665\\
220	0.630237500000001\\
240	0.6494625\\
260	0.662175\\
280	0.679337499999999\\
};
%\addlegendentry{CV}

\end{axis}
\end{tikzpicture}%
\vspace{0.4cm}
\caption{Shrinkage parameter $\beta$  as a function of $n$  when samples are drawn from a  MVN distribution (left panel) and a  $t$-distribution with $ \nu=5$  d.o.f. (right panel), where $\M$ has  an  AR(1) structure; $\varrho=0.6$ and $p=40$.} \label{fig:AR1_beta}
\end{figure}

The data is generated from MVN distribution $\mathcal N_p(\bo 0, \M)$,  where $\M$ has an  AR(1) covariance matrix structure with  $\varrho=0.6$. 
The dimension is $\pdim = 40$ and $\ndim$ varies from 60 to 280.  Value $q=0.7$ determining the threshold $c$ is used  in \eqref{eq:csq} for Huber's weight. 
Since Huber's M-estimator is scaled to be consistent to the covariance matrix for Gaussian samples,  the underlying population parameter $\M_0$  coincides with the covariance matrix $\M$ in this case. We also scaled the MVT-weight $u_{\text{T}}(t; \nu)$ so that it is consistent to $\M$ for Gaussian data. 
\autoref{fig:AR1}  compares the normalized MSE (NMSE) $\| \hat \M_\be - \M \|_{\Fr}^2/\| \M \|_{\Fr}^2$ of different estimators w.r.t. increasing sample length $n$. 
 It can be noted that all estimators provide essentially equally good estimator of the covariance matrix $\M$ for Gaussian data; RSCM and RMVT are performing equally well, largely due to the effect of data-adaptive estimation of d.o.f. parameter $\nu$.  It should be noted that  their performance difference to LW or RHub estimators are still marginal and differences can be spotted only by zooming in as in the sub-plot of  \autoref{fig:AR1}. As expected, RBLW estimator has a slight advantage over the other estimators in this case. The left panel of \autoref{fig:AR1_beta} shows the (average) shrinkage parameter $\beta$ as a function of $n$. As can be noted, the average shrinkage parameter of the proposed RSCM estimator can be seen to be roughly an average of CV and LW shrinkage parameters.

\subsection{Heavy-tailed data} 

Next we computed the NMSE curves when the data is generated from a heavy-tailed $t$-distribution with $\nu= 5$ and $ \nu=3$ d.o.f.   Note that NMSE of each estimator is now compared against the underlying population parameter $\M_0$ of each M-estimator.  \autoref{fig:AR1_t}  displays the results. RBLW  had a very poor performance which is due to its strict assumption of Gaussianity.  It can be noted that CV method performs similarly, but slightly worse, than RSCM or LW. 
This can be partially attested to poor robustness properties of cross-validation.
In the case of $\nu=3$ d.o.f., also the non-robust RSCM and LW provided large NMSE and thus all non-robust estimators are not visible in the right panel of  \autoref{fig:AR1_t}. This was expected since $t$-distribution with $\nu=3$ d.o.f. is very heavy-tailed with non-finite kurtosis. 
As can be noted, the proposed robust RHub and RMVT  estimators  provide significantly improved performance. We can also notice that RMVT estimator that adaptively estimates the d.o.f. $ \nu$ from the data is able to outperform the regularized Huber's estimator (RHub).

The right panel of \autoref{fig:AR1_beta} depicts the (average) shrinkage parameter $\beta$ as a function of $n$ in the case that samples are drawn from a $t$-distribution with $ \nu=5$ d.o.f.  As can be noted the robust shrinkage estimators (RHub and RMVT) use larger shrinkage parameter value $\beta$ than the non-robust RSCM and LW estimators. Compared to RSCM and LW, the RBLW (resp. CV) is  seen to overestimate (resp. underestimate) the shrinkage parameter as it obtains much larger (resp. smaller) values.

Next we investigate how the estimators are able to estimate the shape matrix, i.e., the covariance matrix up to a scale. 
\autoref{fig:AR1shape} displays the  NMSE, $\| \hat{\bo V} - \bo V \|_{\Fr}^2/\| \bo V \|_{\Fr}^2$,   of different shrinkage shape matrix estimators, defined as $ \hat{\bo V} = p \hat \M_\be/\tr(\hat \M_\be)$ when samples are generated  from a $t$-distribution with $ \nu=5$ d.o.f.  
Note that such normalization is not necessary for CWH or RTyl since they verify $\tr(\hat{\bo \M})=p$ in the first place. 
\autoref{fig:AR1shape} illustrate both the case when correlation parameter $\varrho$ of the AR$(1)$ scatter matrix parameter $\M$ is fixed while 
$n$ varies and the case that $n$ is fixed while $\varrho$ varies.
As can be seen from the top panel of \autoref{fig:AR1shape}, all robust shape estimators are performing well and very similarly.  In fact, performance of RMVT and CWH is essentially the same. We can also observe that the two different approaches for shrinking Tyler's M-estimator, so CWH and the proposed RTyl  are very similar. We can note from the bottom panels of \autoref{fig:AR1shape}  that when $\varrho \approx 0$ (so $\M$ is close to a scaled identity matrix) all estimators perform similarly. This is because all estimators are  shrunk heavily towards the scaled identity matrix (namely, $\beta \approx 0$ for all estimators). Similarly, when $\varrho \approx 1$ (so $\M$ is close to a singular matrix of rank $1$), all estimators have a rather similar performance. This is because  the true scatter matrix $\M$ is poorly conditioned ($\mathrm{cond}(\M) \approx7000$) and  all estimators share similar difficulties of capturing the subspace structure  due to limited training data and no {\it a priori} information about such structure. Indeed biggest differences between estimators are observed when $\M$ has no particular structure,  i.e.,  $\varrho$ in the range $[0.4, 0.7]$.

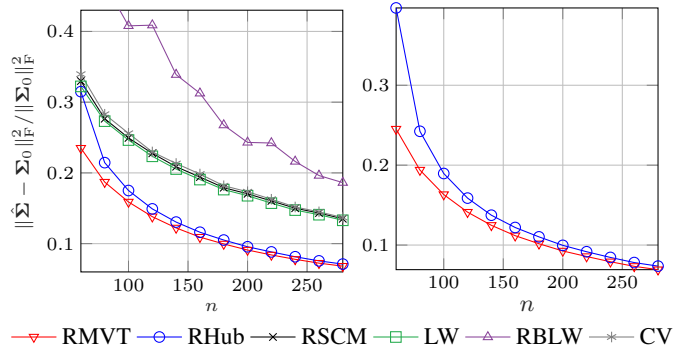
\begin{figure}[!t]
%\vspace{0.2cm} 
\setlength\fwidth{0.40\textwidth}
\definecolor{mycolor1}{rgb}{0.14200,0.67200,0.30300}%
\definecolor{mycolor2}{rgb}{0.54200,0.27200,0.60300}%
\begin{tikzpicture}
\scriptsize
\begin{axis}[%
width=0.48\fwidth,
height=0.48\fwidth,
at={(0\fwidth,0\fwidth)},
scale only axis,
xmin=60,
xmax=280,
xlabel style={font=\color{white!15!black}},
xlabel={$n$},
ymin=0.06,
ymax=0.43,
ylabel style={font=\color{white!15!black}},
ylabel={{$\| \hat \M - \M_0 \|_{\mathrm{F}}^2/\| \M_0 \|_{\mathrm{F}}^2$}},
axis background/.style={fill=white},
xmajorgrids,
ymajorgrids,
legend style={legend columns=-1, font=\small, anchor=north west, legend cell align=right, align=right, fill=none, draw=none, overlay, at={(-0.3,-.17)}}	
]
\addplot [color=red, mark=triangle, mark options={solid, rotate=180, red}] % dashed, 
  table[row sep=crcr]{%
60	0.234989351129897\\
80	0.187021352255102\\
100	0.158904677367105\\
120	0.138121165288443\\
140	0.122201215481695\\
160	0.109208278611978\\
180	0.099556858761056\\
200	0.0909665758652431\\
220	0.0840147707984626\\
240	0.0778155750383841\\
260	0.0724191720284312\\
280	0.0681617236702734\\
};
\addlegendentry{RMVT}

\addplot [color=blue,  mark=o, mark options={solid, blue}] %dashdotted, 
  table[row sep=crcr]{%
60	0.314915514116767\\
80	0.214502523605035\\
100	0.175092133660212\\
120	0.14912473262816\\
140	0.130474152415879\\
160	0.116118529497915\\
180	0.105070390806924\\
200	0.0955635902066864\\
220	0.0882775592118987\\
240	0.0814540108405928\\
260	0.0756267319318625\\
280	0.0710861983838439\\
};
\addlegendentry{RHub}

\addplot [color=black,  mark=x, mark options={solid, black}]
  table[row sep=crcr]{%
60	0.329836216121677\\
80	0.276606330641559\\
100	0.249460446788161\\
120	0.2276754216258\\
140	0.208758446976154\\
160	0.194153624379895\\
180	0.179091285335981\\
200	0.170470252280339\\
220	0.160371652069291\\
240	0.150079571096648\\
260	0.143244827865368\\
280	0.135158330543588\\
};
\addlegendentry{RSCM}

\addplot [color=mycolor1, mark=square, mark options={solid, mycolor1}] %dotted, 
  table[row sep=crcr]{%
60	0.322300321805507\\
80	0.273207725599356\\
100	0.246444105175178\\
120	0.223465968856429\\
140	0.205456102983968\\
160	0.190572944377991\\
180	0.176408435632607\\
200	0.167742172094953\\
220	0.157238554290139\\
240	0.147428154799689\\
260	0.140914069361073\\
280	0.133004103826804\\
};
\addlegendentry{LW}

\addplot [color=mycolor2, mark=triangle, mark options={solid, mycolor2}] %dotted, 
  table[row sep=crcr]{%
60	0.961995438746567\\
80	0.468000893055517\\
100	0.407927972931537\\
120	0.408873889061668\\
140	0.338903553693445\\
160	0.312542116664462\\
180	0.267597099126292\\
200	0.243127794254307\\
220	0.242240089676987\\
240	0.216167275357581\\
260	0.196291501731024\\
280	0.186265509490257\\
};
\addlegendentry{RBLW}

\addplot [color=gray,  mark=asterisk, mark options={solid, gray}] %dashed, 
  table[row sep=crcr]{%
60	0.339537646492296\\
80	0.282937410871142\\
100	0.255998102336585\\
120	0.229653897385105\\
140	0.2133126690287\\
160	0.198387901829072\\
180	0.181842771898667\\
200	0.173369615616503\\
220	0.162770684763515\\
240	0.151812351029613\\
260	0.144953337639582\\
280	0.136924285748437\\
};
\addlegendentry{CV}

\end{axis}

\end{tikzpicture}% 
\begin{tikzpicture} %[font=\Large,scale=0.6]
\scriptsize
\begin{axis}[%
width=0.48\fwidth,
height=0.48\fwidth,
scale only axis,
xmin=60,
xmax=280,
xlabel style={font=\color{white!15!black}},
xlabel={{\small $n$}},
ymin=0.0692150515242403,
ymax=0.396758124125362,
axis background/.style={fill=white},
xmajorgrids,
ymajorgrids,
legend style={legend columns=1, legend cell align=left, align=left, fill=none, draw=none}
]
\addplot [color=red, mark=triangle, mark options={solid, rotate=180, red}] %dashed, 
  table[row sep=crcr]{%
60	0.245166813816713\\
80	0.193902388442634\\
100	0.163220901296709\\
120	0.141198815242516\\
140	0.124716934329047\\
160	0.111731472949592\\
180	0.101659609495418\\
200	0.0926143349381351\\
220	0.0856737136133183\\
240	0.0791200919866976\\
260	0.0733594521211031\\
280	0.0692150515242403\\
};
%\addlegendentry{RMVT}

\addplot [color=blue,  mark=o, mark options={solid, blue}] %dashdotted,
  table[row sep=crcr]{%
60	0.396758124125362\\
80	0.242386059471335\\
100	0.189560763041202\\
120	0.158766123307292\\
140	0.137365609628373\\
160	0.121977958349076\\
180	0.110239338056141\\
200	0.0996228886467904\\
220	0.0916335152263573\\
240	0.0845104394308717\\
260	0.0780973965453987\\
280	0.0736146543500461\\
};
%\addlegendentry{RHub}

\end{axis}
\end{tikzpicture}%
\vspace{0.4cm}
\caption{NMSE  as a  function of $n$  when samples are drawn from a $p$-variate $t$-distribution with $\nu=5$ (left panel) and $\nu=3$ (right panel) d.o.f. The scatter matrix follows an AR(1) structure; $\varrho=0.6$ and $p=40$.} \label{fig:AR1_t}
\end{figure}

\begin{figure}[!t]
\setlength\fwidth{0.58\textwidth}
\centering
\subfloat[$\varrho=0.6$]{% This file was created by matlab2tikz.
\definecolor{mycolor1}{rgb}{0.14200,0.67200,0.30300}%
\definecolor{mycolor2}{rgb}{0.54200,0.27200,0.60300}%
\definecolor{mycolor3}{rgb}{0.40000,0.40000,0.00000}%
\begin{tikzpicture}

\begin{axis}[%
width=0.56\fwidth,
height=0.44\fwidth,
scale only axis,
xmin=60,
xmax=280,
xlabel style={font=\color{white!15!black}},
xlabel={\small $n$},
ymin=0.0660134542544466,
ymax=0.440950550589464,
ylabel style={font=\color{white!15!black}},
ylabel={\small $\| \hat{ \bo V}- \bo V \|_{\mathrm{F}}^2/\| \bo V \|_{\mathrm{F}}^2$},
axis background/.style={fill=white},
xmajorgrids,
ymajorgrids,
%legend style={legend columns=4, font=\small, anchor=north west, legend cell align=left, align=right, fill=none, draw=none,  at={(-0.08,-.17)}}	% this for legend at bottom 
legend style={legend columns=4, font=\small, anchor=north west, legend cell align=left, align=right, fill=none, draw=none,  at={(-0.08,1.19)}}	
]
\addplot [color=red, line width=0.7pt, mark size=2.0pt, mark=triangle, mark options={solid, rotate=180, red}] %dashed, 
  table[row sep=crcr]{%
60	0.218568297342243\\
80	0.178764986455592\\
100	0.152725405216209\\
120	0.132889035187492\\
140	0.117885864361217\\
160	0.105627699046592\\
180	0.0962945028035744\\
200	0.0881443328799346\\
220	0.0812107042229612\\
240	0.0753549710826732\\
260	0.0701605928004257\\
280	0.0660134542544466\\
};
\addlegendentry{RMVT}

\addplot [color=blue,  line width=0.7pt,mark size=2.0pt, mark=o, mark options={solid, blue}] %dashdotted,
  table[row sep=crcr]{%
60	0.228775463413582\\
80	0.185124836212433\\
100	0.157660280897378\\
120	0.136989743999578\\
140	0.12154145817122\\
160	0.108851260824072\\
180	0.0991367863231486\\
200	0.0907704838421474\\
220	0.0836843444273323\\
240	0.0776138797496422\\
260	0.0721737100003885\\
280	0.0679594449304121\\
};
\addlegendentry{RHub}

\addplot [color=black,line width=0.7pt, mark size=2.0pt,  mark=x, mark options={solid, black}]
  table[row sep=crcr]{%
60	0.289929808444324\\
80	0.256405770956407\\
100	0.232296032302098\\
120	0.209685721577823\\
140	0.193993649078867\\
160	0.180855360852838\\
180	0.167243314280258\\
200	0.159295498264014\\
220	0.149245783026175\\
240	0.140253038741189\\
260	0.134313496727826\\
280	0.127014461123792\\
};
\addlegendentry{RSCM}

\addplot [color=mycolor1,line width=0.7pt,mark size=2.0pt, mark=square, mark options={solid, mycolor1}] %dotted, 
  table[row sep=crcr]{%
60	0.291212070004899\\
80	0.257149433235268\\
100	0.232822467790705\\
120	0.209897174252829\\
140	0.194053257076012\\
160	0.180866910867135\\
180	0.167320967396402\\
200	0.159103989092202\\
220	0.148982290296682\\
240	0.140055833054648\\
260	0.134059657306949\\
280	0.126844121260439\\
};
\addlegendentry{LW}

\addplot [color=mycolor2, line width=0.7pt, mark size=2.0pt, mark=triangle, mark options={solid, mycolor2}] %dotted, 
  table[row sep=crcr]{%
60	0.440950550589464\\
80	0.363590441971697\\
100	0.314835587172849\\
120	0.276579995100958\\
140	0.254344694626137\\
160	0.233124934067971\\
180	0.210245447644158\\
200	0.200852611939607\\
220	0.188382370687852\\
240	0.172483571617598\\
260	0.16412349363849\\
280	0.153750089045556\\
};
\addlegendentry{RBLW}

\addplot [color=gray, line width=0.7pt, mark size=2.0pt, mark=asterisk, mark options={solid, gray}] %dashed,
  table[row sep=crcr]{%
60	0.302633147904937\\
80	0.267764389493273\\
100	0.243172373440151\\
120	0.219187849095752\\
140	0.202575158862016\\
160	0.188772626605334\\
180	0.174157605804473\\
200	0.165880311647368\\
220	0.155054097777486\\
240	0.145711857959601\\
260	0.139439605402387\\
280	0.131844410767306\\
};
\addlegendentry{CV}

\addplot [color=mycolor3, line width=0.7pt, mark size=2.0pt, mark=+, mark options={solid, mycolor3}]
  table[row sep=crcr]{%
60	0.222106339183032\\
80	0.180629134120674\\
100	0.154000466959643\\
120	0.13389969588509\\
140	0.118669032171595\\
160	0.106282126881315\\
180	0.0969037318336481\\
200	0.0886842738748952\\
220	0.0816666728262332\\
240	0.0757817303608897\\
260	0.0705847244044556\\
280	0.0663976729859238\\
};
\addlegendentry{RTyl}

\addplot [color=black!60!blue, line width=0.7pt, mark size=2.0pt, mark=triangle, mark options={solid, rotate=90, black!60!blue}] %dashdotted, 
  table[row sep=crcr]{%
60	0.21556546749848\\
80	0.17822460028166\\
100	0.152583333392534\\
120	0.133009801594739\\
140	0.118034122910117\\
160	0.105785377414094\\
180	0.0965441432828716\\
200	0.0883660197369836\\
220	0.0814271971880657\\
240	0.0755879089195379\\
260	0.0704107035068735\\
280	0.0662479132098619\\
};
\addlegendentry{CWH}

\end{axis}

\begin{axis}[%
width=0.23\fwidth,
height=0.20\fwidth,
at={(0.302\fwidth,0.22\fwidth)},
scale only axis,
xmin=99,
xmax=101,
xtick={ 99, 100, 101},
ymin=0.151,
ymax=0.159,
 yticklabel style={/pgf/number format/.cd,fixed,precision=3},
axis background/.style={fill=white},
legend style={legend cell align=left, align=left, draw=white!15!black}
]
\addplot [color=red, line width=0.9pt, mark size=2.0pt,   mark=triangle, mark options={solid, rotate=180, red}] %dashed,
  table[row sep=crcr]{%
80	0.178764986455592\\
100	0.152725405216209\\
120	0.132889035187492\\
};

\addplot [color=blue, line width=0.9pt, mark size=2.0pt,  mark=o, mark options={solid, blue}] %dashdotted, 
  table[row sep=crcr]{%
80	0.185124836212433\\
100	0.157660280897378\\
120	0.136989743999578\\
};

\addplot [color=mycolor3,line width=0.9pt, mark size=2.0pt,   mark=+, mark options={solid, mycolor3}]
  table[row sep=crcr]{%
80	0.180629134120674\\
100	0.154000466959643\\
120	0.13389969588509\\
};

\addplot [color=black!50!blue,  line width=0.9pt, mark size=2.0pt,   mark=triangle, mark options={solid, rotate=90, black!50!blue}] %dashdotted,
  table[row sep=crcr]{%
80	0.17822460028166\\
100	0.152583333392534\\
120	0.133009801594739\\
};

\end{axis}
\end{tikzpicture}%
\label{fig:AR1shape_a}}
\setlength\fwidth{0.40\textwidth}
\subfloat[$n=60$]{% This file was created by matlab2tikz.
\definecolor{mycolor1}{rgb}{0.14200,0.67200,0.30300}%
\definecolor{mycolor2}{rgb}{0.40000,0.40000,0.00000}%
\begin{tikzpicture}
\scriptsize
\begin{axis}[%
width=0.48\fwidth,
height=0.48\fwidth,
at={(0\fwidth,0\fwidth)},
scale only axis,
xmin=0.01,
xmax=0.99,
xlabel style={font=\color{white!15!black}},
xlabel={\small $\rho$},
ymin=0.000504499890995826,
ymax=0.302391223125055,
ylabel style={font=\color{white!15!black}},
ylabel={$\| \hat{ \bo V}- \bo V \|_{\mathrm{F}}^2/\| \bo V \|_{\mathrm{F}}^2$},
axis background/.style={fill=white},
xmajorgrids,
ymajorgrids,
]
\addplot [color=red, mark options={solid, rotate=180, red}]
  table[row sep=crcr]{%
0.01	0.000725102387830098\\
0.1	0.0196157506410019\\
0.2	0.0695919001610987\\
0.3	0.130611059396178\\
0.4	0.183731025546231\\
0.5	0.21499961916004\\
0.6	0.218568297342243\\
0.7	0.195275028980625\\
0.8	0.149440639932588\\
0.9	0.0851088280839941\\
0.99	0.00901925113936942\\
};
%\addlegendentry{RMVT}

\addplot [color=blue,mark=o, mark options={solid, blue}]
  table[row sep=crcr]{%
0.01	0.000607075888557948\\
0.1	0.0194756223613268\\
0.2	0.069833555961607\\
0.3	0.132447663858524\\
0.4	0.188240990282302\\
0.5	0.222630921571338\\
0.6	0.228775463413582\\
0.7	0.206591793321176\\
0.8	0.159814329866107\\
0.9	0.0921251115766955\\
0.99	0.0104884310217124\\
};
%\addlegendentry{RHub}

\addplot [color=black, mark=x, mark options={solid, black}]
  table[row sep=crcr]{%
0.01	0.000504499890995826\\
0.1	0.0194880493742978\\
0.2	0.0717168964208952\\
0.3	0.141770202098149\\
0.4	0.211980651222111\\
0.5	0.265489316908488\\
0.6	0.289929808444324\\
0.7	0.278479536980706\\
0.8	0.228886575500452\\
0.9	0.140335714400963\\
0.99	0.0204917756316276\\
};
%\addlegendentry{RSCM}

\addplot [color=mycolor1, mark=square, mark options={solid, mycolor1}]
  table[row sep=crcr]{%
0.01	0.00360646834842675\\
0.1	0.02246653628072\\
0.2	0.0741844062661194\\
0.3	0.144001905760762\\
0.4	0.213849005303455\\
0.5	0.266992678894246\\
0.6	0.291212070004899\\
0.7	0.279716270536932\\
0.8	0.229982658956798\\
0.9	0.140256538688018\\
0.99	0.0186759024168597\\
};
%\addlegendentry{LW}

\addplot [color=gray, mark=asterisk, mark options={solid, gray}]
  table[row sep=crcr]{%
0.01	0.000995073032543253\\
0.1	0.0198403297314513\\
0.2	0.0722032102751683\\
0.3	0.143582948617177\\
0.4	0.216121515570968\\
0.5	0.273279198358849\\
0.6	0.302391223125055\\
0.7	0.296597533652204\\
0.8	0.253377142827457\\
0.9	0.173650383566424\\
0.99	0.0693072511914826\\
};
%\addlegendentry{CV}

\addplot [color=mycolor2,  mark=+, mark options={solid, mycolor2}]
  table[row sep=crcr]{%
0.01	0.000748482615194178\\
0.1	0.0196624926136079\\
0.2	0.0697961913590655\\
0.3	0.131350034423155\\
0.4	0.185419494248587\\
0.5	0.217752856248069\\
0.6	0.222106339183032\\
0.7	0.199021607948765\\
0.8	0.152671279147201\\
0.9	0.0870030697536966\\
0.99	0.0088542961484244\\
};
%\addlegendentry{RTyl}

\addplot [color=black!70!blue, mark=triangle, mark options={solid, rotate=90, black!70!blue}]
  table[row sep=crcr]{%
0.01	0.000700325627522554\\
0.1	0.0195709613837802\\
0.2	0.0693700585096262\\
0.3	0.129535591808867\\
0.4	0.181520944008636\\
0.5	0.212278892204802\\
0.6	0.21556546749848\\
0.7	0.191678313304104\\
0.8	0.145240366360079\\
0.9	0.081021175133647\\
0.99	0.0068435131824872\\
};
%\addlegendentry{CWH}

\end{axis}
\end{tikzpicture}%
\label{fig:AR1shape_b}}
\subfloat[$n=120$]{% This file was created by matlab2tikz.
\definecolor{mycolor1}{rgb}{0.14200,0.67200,0.30300}%
\definecolor{mycolor2}{rgb}{0.40000,0.40000,0.00000}%
\begin{tikzpicture}
\scriptsize
\begin{axis}[%
width=0.48\fwidth,
height=0.48\fwidth,
at={(0\fwidth,0\fwidth)},
scale only axis,
xmin=0.01,
xmax=0.99,
xlabel style={font=\color{white!15!black}},
xlabel={\small $\rho$},
ymin=0.00036272596943515,
ymax=0.219037924853015,
 yticklabel style={/pgf/number format/.cd,fixed,precision=2},
axis background/.style={fill=white},
xmajorgrids,
ymajorgrids,
]
\addplot [color=red, mark=triangle, mark options={solid, rotate=180, red}]
  table[row sep=crcr]{%
0.01	0.000505120936454884\\
0.1	0.0189176002108171\\
0.2	0.062155344046241\\
0.3	0.106276160227824\\
0.4	0.134785214841453\\
0.5	0.142793170329264\\
0.6	0.132889035187492\\
0.7	0.110095974213911\\
0.8	0.079003129571135\\
0.9	0.0421330141534193\\
0.99	0.00354113901803021\\
};
%\addlegendentry{RMVT}

\addplot [color=blue, mark=o, mark options={solid, blue}]
  table[row sep=crcr]{%
0.01	0.000469311947090165\\
0.1	0.0188930391883767\\
0.2	0.0625702970526931\\
0.3	0.107767116675541\\
0.4	0.13757537318512\\
0.5	0.146552976705113\\
0.6	0.136989743999578\\
0.7	0.113886177571156\\
0.8	0.0819466382596226\\
0.9	0.0438288410384132\\
0.99	0.00380112306442088\\
};
%\addlegendentry{RHub}

\addplot [color=black, mark=x, mark options={solid, black}]
  table[row sep=crcr]{%
0.01	0.00036272596943515\\
0.1	0.0191025576635515\\
0.2	0.0676716172726876\\
0.3	0.127006477599916\\
0.4	0.177741572872004\\
0.5	0.206914516396337\\
0.6	0.209685721577823\\
0.7	0.187133191394226\\
0.8	0.14299309847989\\
0.9	0.0802440555333465\\
0.99	0.00888030461526212\\
};
%\addlegendentry{RSCM}

\addplot [color=mycolor1, mark=square, mark options={solid, mycolor1}]
  table[row sep=crcr]{%
0.01	0.00114103094821682\\
0.1	0.0198470917059182\\
0.2	0.0686022589405766\\
0.3	0.127838094890577\\
0.4	0.17837421299539\\
0.5	0.207321149870929\\
0.6	0.209897174252829\\
0.7	0.187193478131448\\
0.8	0.142869581165313\\
0.9	0.079687212366745\\
0.99	0.00805853358597888\\
};
%\addlegendentry{LW}

\addplot [color=gray, mark=asterisk, mark options={solid, gray}]
  table[row sep=crcr]{%
0.01	0.000624708012850996\\
0.1	0.0193224975517778\\
0.2	0.0686009533489917\\
0.3	0.129309550819016\\
0.4	0.182191947250933\\
0.5	0.213911396702749\\
0.6	0.219037924853015\\
0.7	0.198617431686971\\
0.8	0.156591175020703\\
0.9	0.0976549147951894\\
0.99	0.0358747379022347\\
};
%\addlegendentry{CV}

\addplot [color=mycolor2, mark=+, mark options={solid, mycolor2}]
  table[row sep=crcr]{%
0.01	0.000512074408790379\\
0.1	0.018937860073575\\
0.2	0.0622707409488322\\
0.3	0.106643558425067\\
0.4	0.135469139559923\\
0.5	0.143716992739195\\
0.6	0.13389969588509\\
0.7	0.111034820728456\\
0.8	0.0797439869392675\\
0.9	0.0425743700161754\\
0.99	0.00356248599188928\\
};
%\addlegendentry{RTyl}

\addplot [color=black!70!blue,  mark=triangle, mark options={solid, rotate=90, black!70!blue}]
  table[row sep=crcr]{%
0.01	0.000490721441547201\\
0.1	0.0188996121838204\\
0.2	0.062094093616778\\
0.3	0.106053831767406\\
0.4	0.134628244318899\\
0.5	0.142840972671412\\
0.6	0.133009801594739\\
0.7	0.110121892813474\\
0.8	0.0788706389559361\\
0.9	0.0418642588810951\\
0.99	0.00318987280649322\\
};
%\addlegendentry{CWH}

\end{axis}
\end{tikzpicture}%
\label{fig:AR1shape_c}} 
\caption{ NMSE of different shrinkage estimators of shape matrix $\bo V$  when samples are drawn a $p=40$ variate $t$-distribution with  $\nu=5$ d.o.f. having  an AR(1)  structure: (a) $\varrho=0.6$  and  sample length $n$ varies; (b) and (c) illustrate the case when $\varrho$ varies  while the  sample length $n$ is fixed.} \label{fig:AR1shape}
\end{figure}
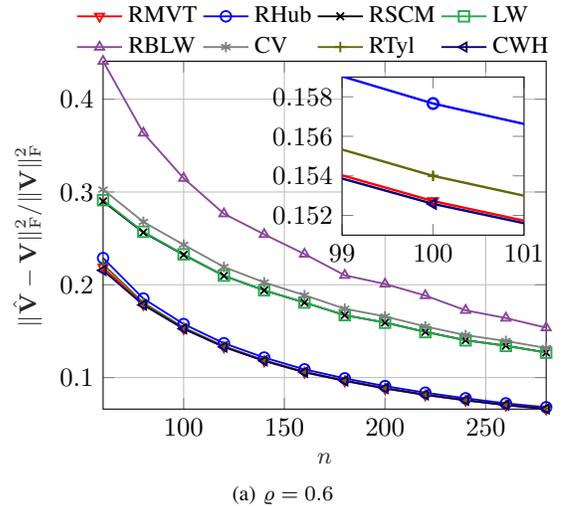
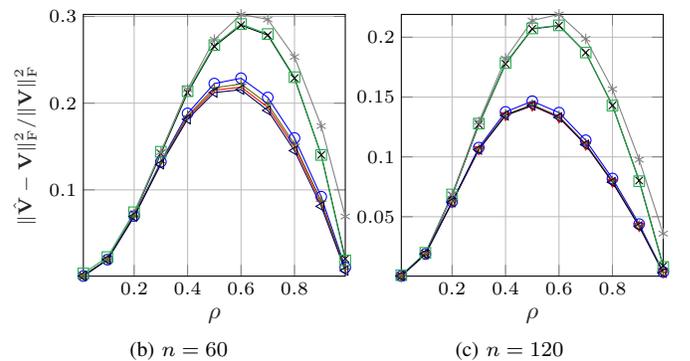

\subsection{Complex-valued data} 

Finally, we note that an important property of our shrinkage method is that it can be used for complex-valued data as well. Some other methods in the previous study, such as RBLW or LW assume real-valued data. In the supplementary material, we provide the results of a simulation study in the same set-up, but now the data being generated from circular complex Gaussian and heavy-tailed $t$-distribution, respectively. These are distributions in the class of CES distributions. In our study we also include empirical Bayes diagonal loading estimator (EBDL) \cite{coluccia2015regularized} which was developed for complex circular Gaussian data. Results obtained for complex-valued data attest the validity of the findings in the real-valued case.

%%%%%%%%%
% SECTION 6
%%%%%%%%%

\section{Application to financial data and portfolio design}  \label{sec:applic}

\subsection{Financial data}
We use the S\&P 500 stock market index (see \autoref{fig:SP500-prices}), which measures the stock performance of 500 large companies listed on stock exchanges in the United States, and its constituent stocks during the period 2016-01-01 to 2020-01-31.

We can easily observe from the returns shown in \autoref{fig:SP500-returns} the effect of volatility clustering over time that is responsible for heavy tails. More concretely, if we assume that the data follows an MVT distribution, then we can compute the degrees of freedom $\nu$ on a rolling-window basis and verify that indeed the data has heavy tails with $\nu \approx 5$ (from mid-2017, $\nu$ varies between $4.5$ and $6$) as shown in \autoref{fig:SP500-nu}.

\begin{figure}[!t]
\centerline{\includegraphics[width=0.5\textwidth]{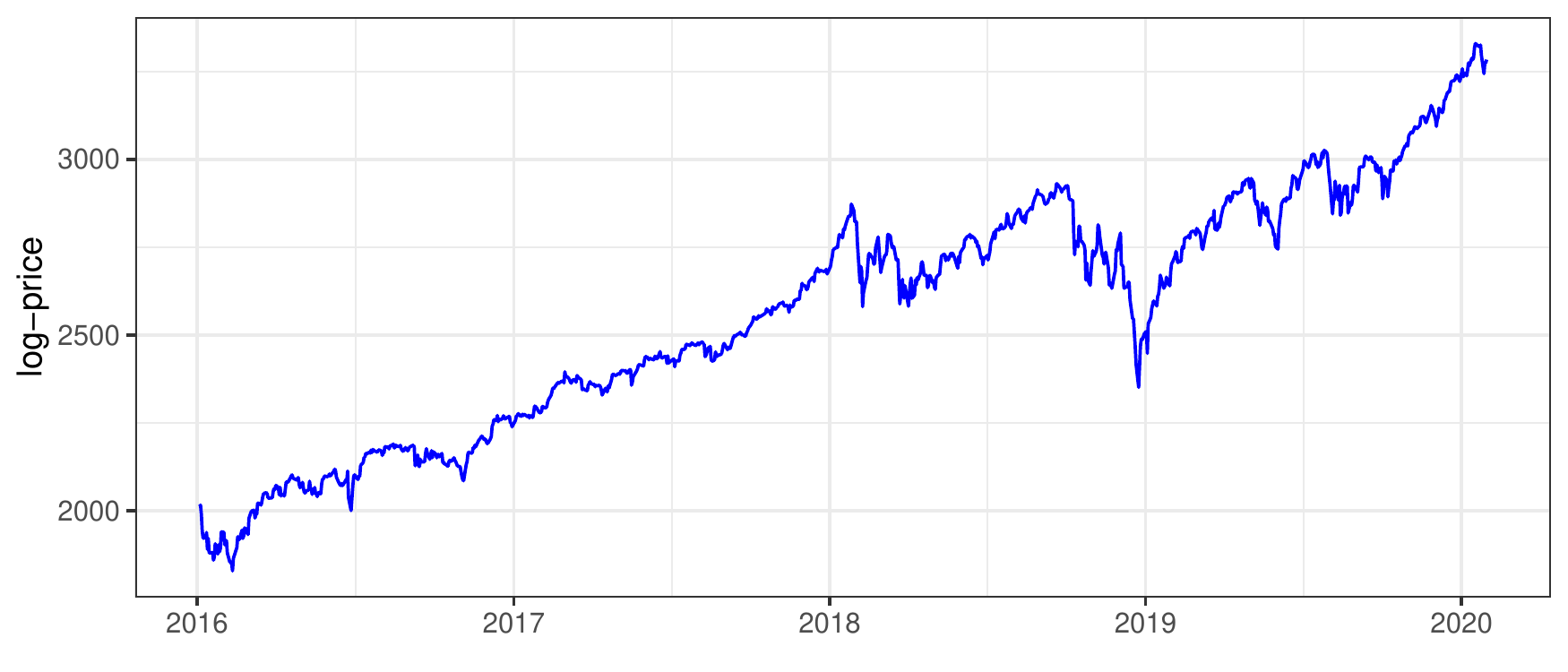}}
\caption{Log-prices of the  S\&P 500 index.} \label{fig:SP500-prices}
\end{figure}

\begin{figure}[!t]
\centerline{\includegraphics[width=0.5\textwidth]{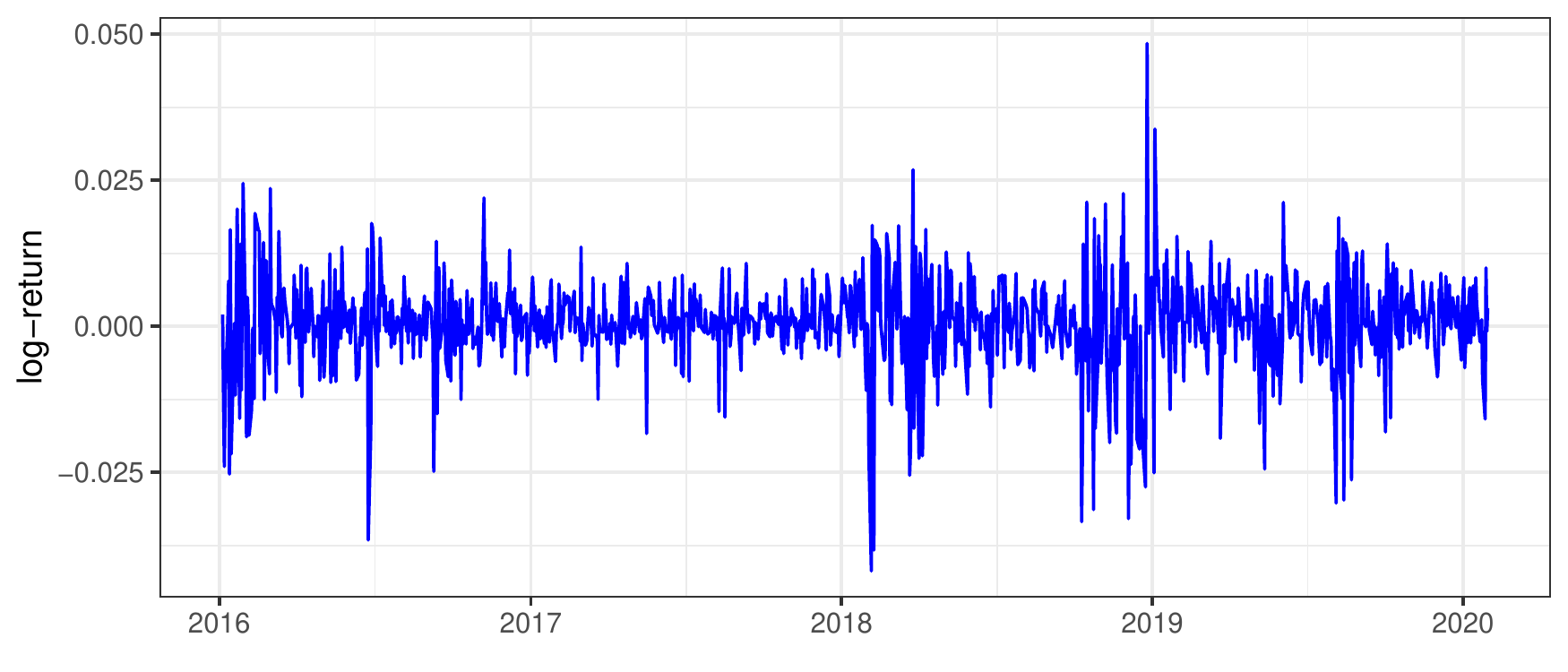}}
\caption{Log-returns of the  S\&P 500 index.} \label{fig:SP500-returns}
\end{figure}

\begin{figure}[!t]
\centerline{\includegraphics[width=0.5\textwidth]{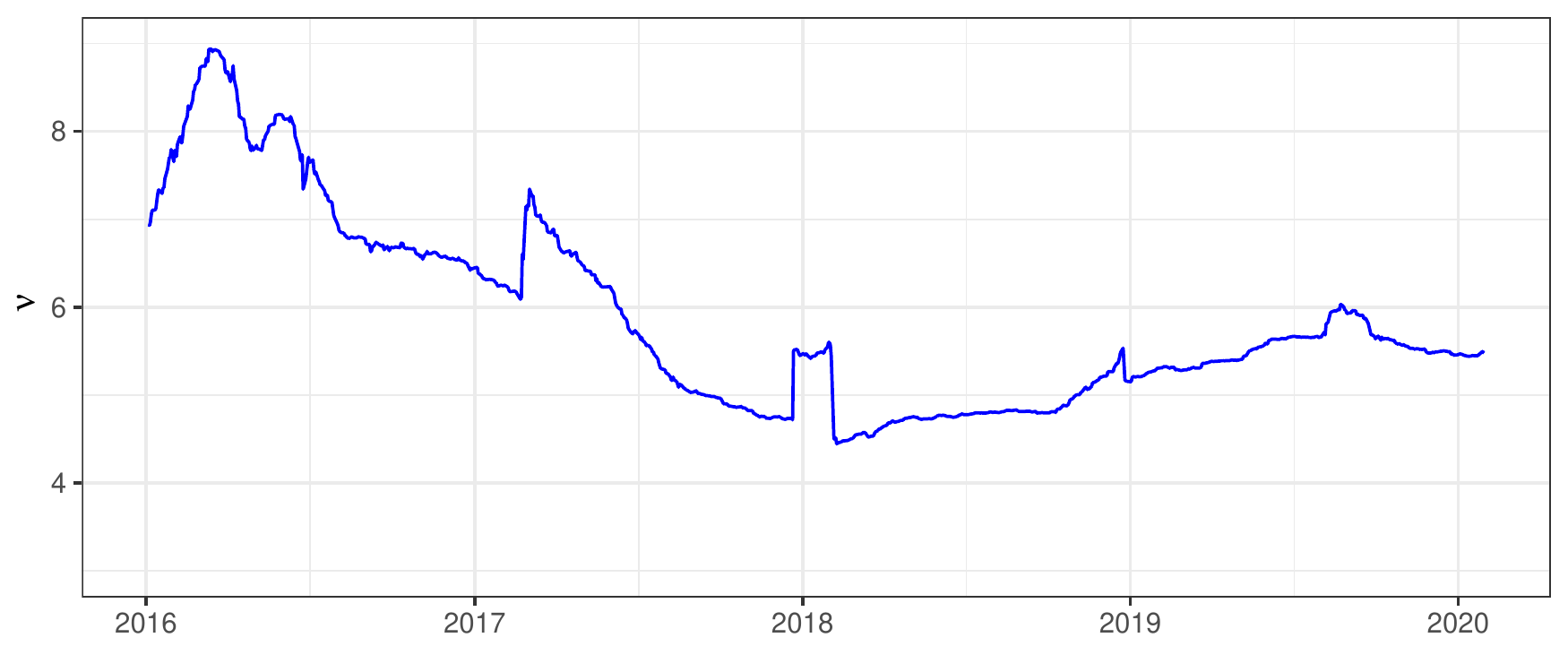}}
\caption{Degrees of freedom $\nu$ of the  S\&P 500 index.} \label{fig:SP500-nu}
\end{figure}

\textbf{Factor model.} Let $\mathbf{x}_i$ denote the returns of the $p$ stocks at time $i$. It turns out that the stock returns are largely driven by very few $k \ll p$ financial factors $\mathbf{f}_i$ as
\begin{equation}
\mathbf{x}_i = \mathbf{B}\mathbf{f}_i + \boldsymbol{\epsilon}_i,
\end{equation}
where $\mathbf{B}\in\mathbb{R}^{p\times k}$ is the factor loading matrix (which is very tall since $k \ll p$) and $\boldsymbol{\epsilon}_i$ is the residual idiosyncratic component. As a consequence, the covariance matrix of these data has the form (assuming normalized factors):
\begin{equation}
\boldsymbol{\Sigma} = \mathbf{B}\mathbf{B}^\top + \boldsymbol{\Psi}
\end{equation}
where $\boldsymbol{\Psi}$ is the (diagonal) covariance matrix of the residuals. Typically, the term $\mathbf{B}\mathbf{B}^\top$ is much stronger than $\boldsymbol{\Psi}$ and this leads to the commonly used spike model in RMT (Random Matrix Theory) \cite{BunBouchaudPotters2016} which contains a few large eigenvalues and the rest small eigenvalues form the so-called bulk.  \autoref{fig:histogram-eigenvalues} shows the histogram of the empirical eigenvalues of the covariance matrix estimated from $p=50$ stocks from the S\&P 500 market data, where a very strong eigenvalue can be observed corresponding to the market factor.

\begin{figure}[!t]
\centerline{\includegraphics[width=0.5\textwidth]{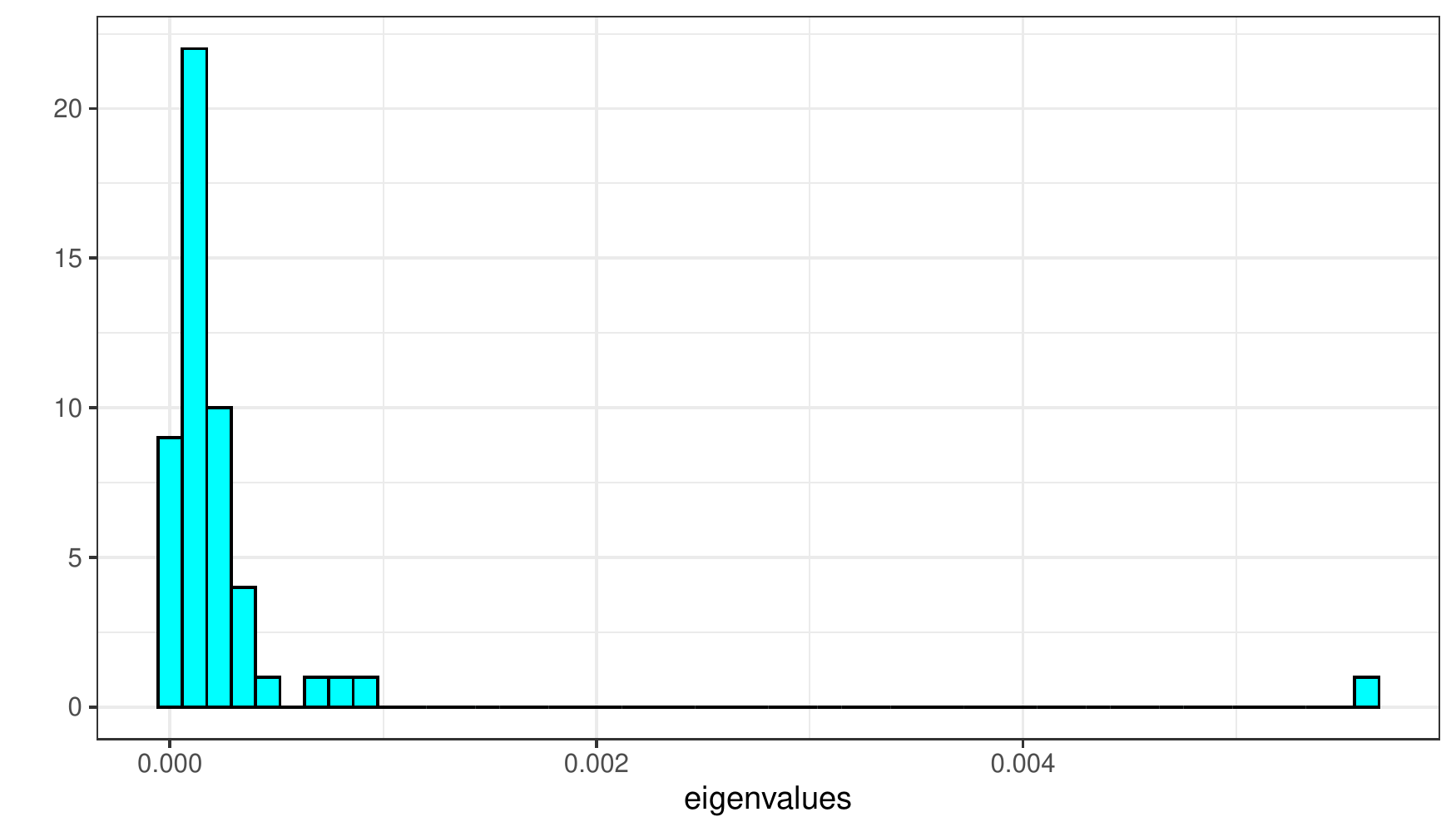}}
\caption{Histogram of empirical eigenvalues obtained from the market data ($p=50$) showing a strong market factor.} \label{fig:histogram-eigenvalues}
\end{figure}

\subsection{Results in terms of MSE}
Since our shrinkage estimators are derived to minimize the MSE in the estimation of the covariance matrix, we start by showing the obtained MSE in the context of financial data. We consider seven methods in our comparison:
\begin{itemize}
  \item LWE is the Ledoit-Wolf estimator \cite{ledoit2004well};
  \item RMVT-Ell1 described in \autoref{subsect:tMLE};
  \item MVT: equals RMVT-Ell1  with no shrinkage ($\beta=1$);
  \item RSCM-Ell1 estimator described in \autoref{subsec:RSCM};
  \item SCM:  equals RSCM-Ell1  with no shrinkage ($\beta=1$);
  \item RHub-Ell1 described in \autoref{subsect:huber}; and
  \item RSCM-Ell2, RMVT-Ell2, and RHub-Ell2 are as RSCM-Ell1, RMVT-Ell1, and RHub-Ell1, respectively, but using $\hat \gamma^{\text{Ell2}}$ estimator of sphericity $\gamma$ ({\it cf.} \autoref{subsec:pract_comp}).
\end{itemize}
 
To make sure that the robust estimators do not underperform the benchmarks when the data is not heavy-tailed, we start by generating synthetic MVN data following the empirical covariance matrix previously obtained from market data (see \autoref{fig:histogram-eigenvalues}). \autoref{fig:MSE-covmat-Gaussian} displays the normalized MSE of the covariance matrix $\mathbb{E}\big[\|\hat{\boldsymbol{\Sigma}} - \boldsymbol{\Sigma}\|_\Fr^2\big]$ via 200 Monte-Carlo simulations. We do not observe any significant difference among the methods.  \autoref{fig:MSE-precmat-Gaussian} shows the normalized MSE of the precision matrix $\mathbb{E}\big[\|\hat{\boldsymbol{\Sigma}}^{-1} - \boldsymbol{\Sigma}^{-1}\|_\Fr^2\big]$ via 200 Monte-Carlo simulations. The main observation is that the two methods without shrinkage significantly underperform. 

\begin{figure}[!t]
\centerline{\includegraphics[width=0.5\textwidth]{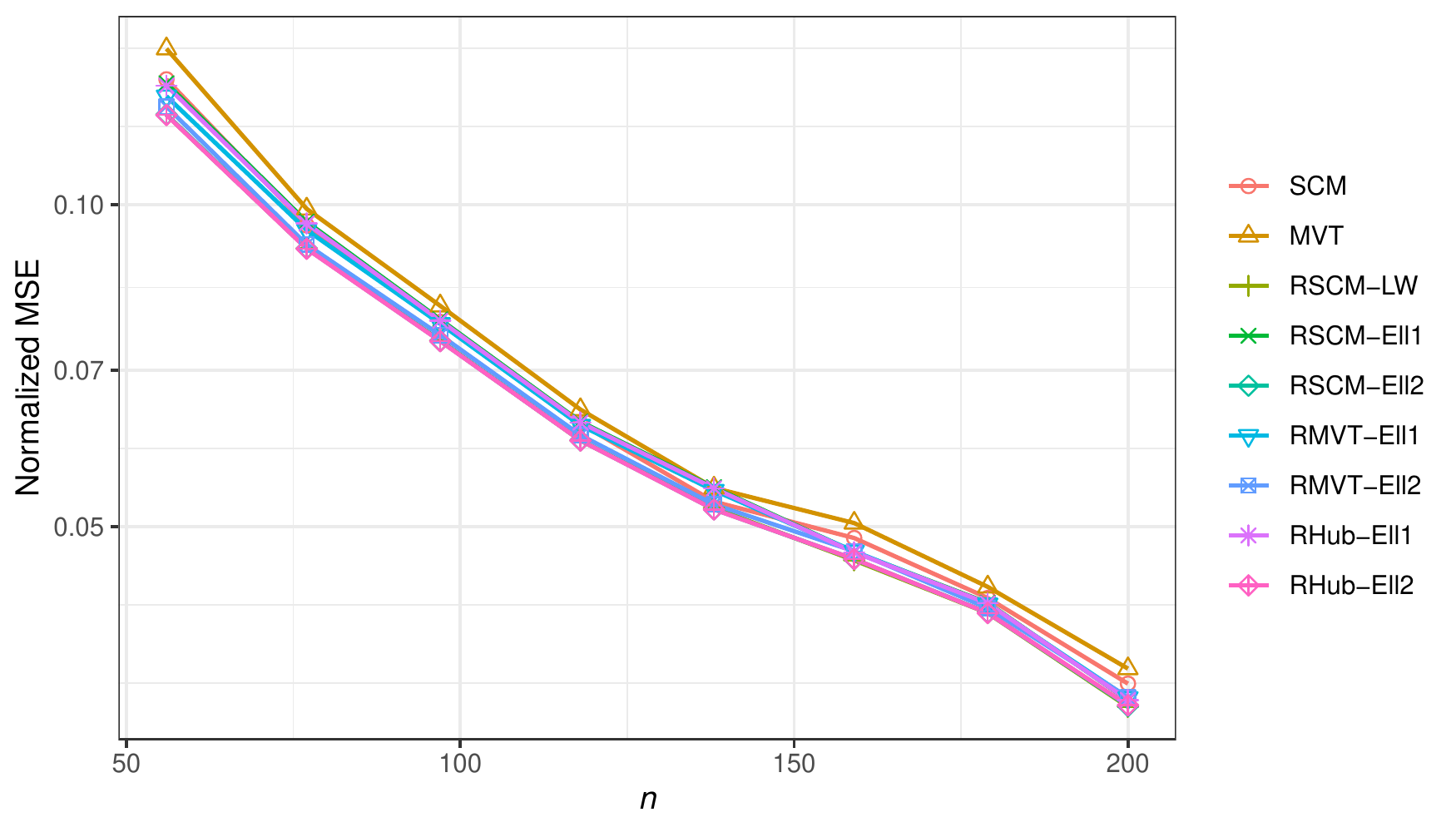}}
\caption{Normalized MSE of covariance matrix vs number of observations for Gaussian data (with $p=50$).} \label{fig:MSE-covmat-Gaussian}
\end{figure}

\begin{figure}[!t]
\centerline{\includegraphics[width=0.5\textwidth]{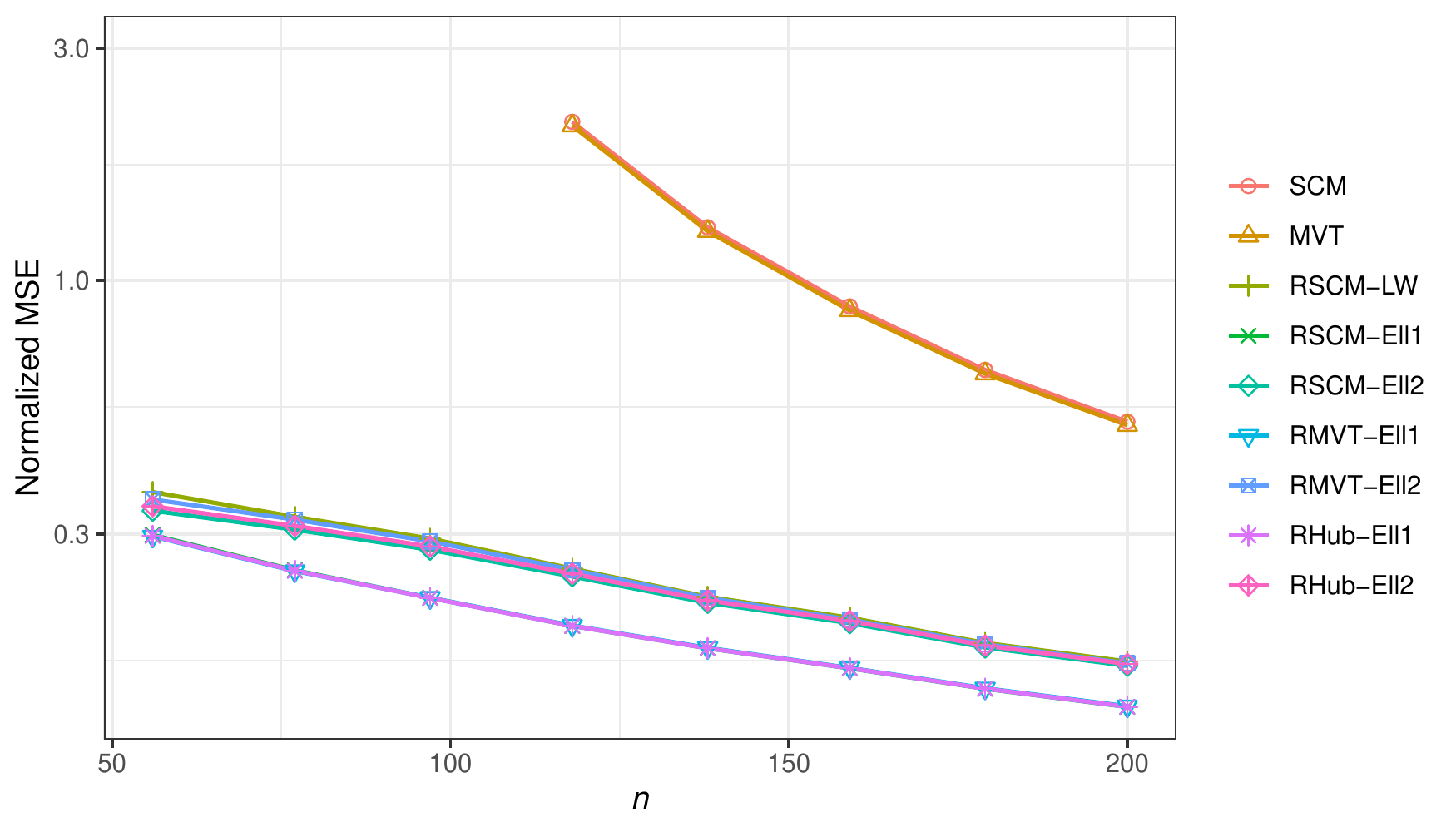}}
\caption{Normalized MSE of precision matrix vs number of observations for Gaussian data (with $p=50$).} \label{fig:MSE-precmat-Gaussian}
\end{figure}

We now generate heavy-tailed synthetic data from a $t$-distribution following the empirical covariance matrix previously obtained from market data (see \autoref{fig:histogram-eigenvalues}) and with d.o.f. $\nu = 5$. \autoref{fig:MSE-covmat-t} shows the normalized MSE of the covariance matrix. We can clearly observe a striking difference showing the superiority of robust estimators (i.e.,
MVT, RMVT, RHub) over non-robust estimators. Among the robust estimators, the proposed shrinkage methods clearly outperform the non-shrinkage
MVT in the low-sample regime, as expected.  \autoref{fig:MSE-precmat-t} displays the normalized MSE of the precision matrix, with similar observations. 

\begin{figure}[!t]
\centerline{\includegraphics[width=0.5\textwidth]{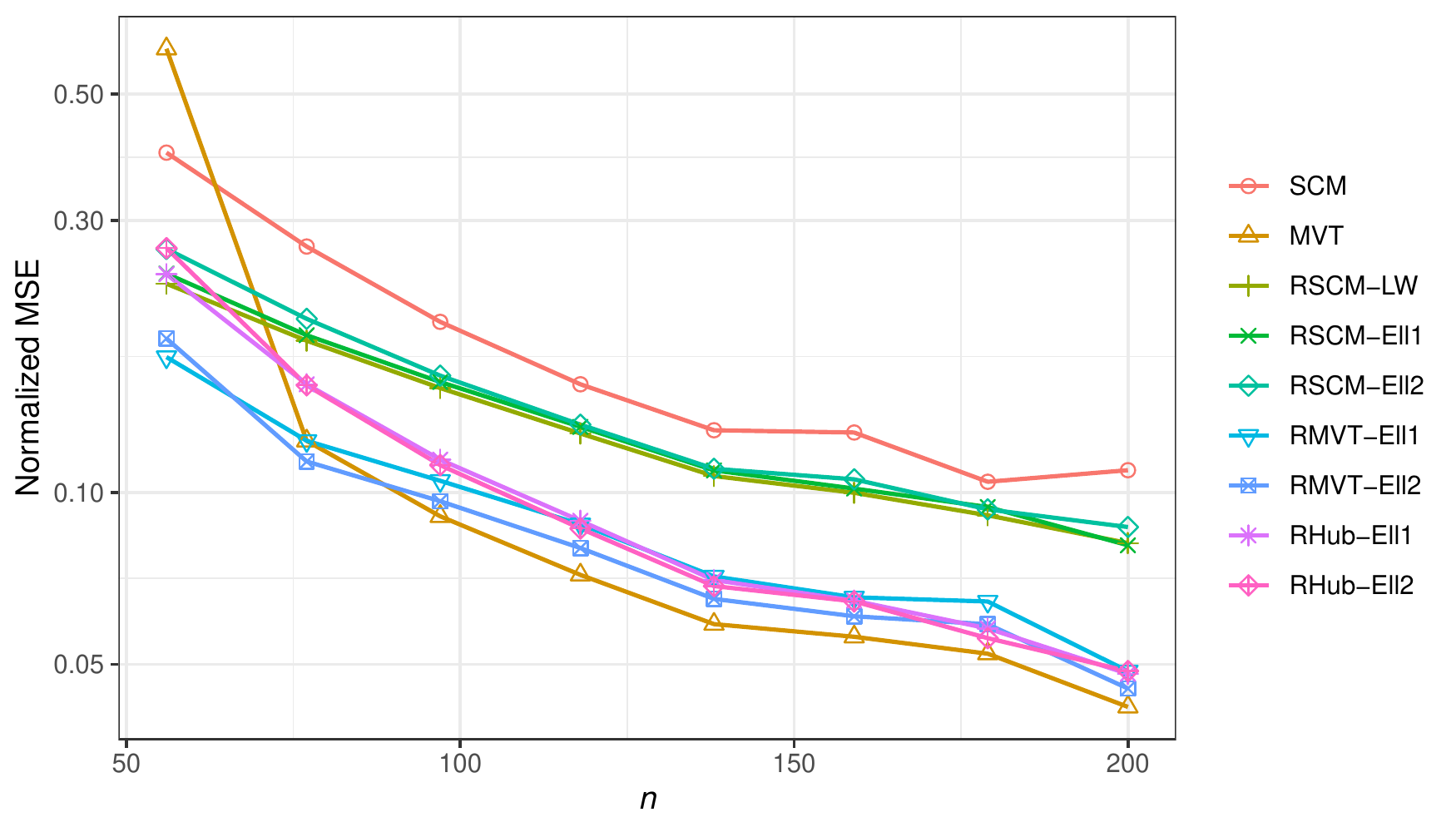}}
\caption{Normalized MSE of covariance matrix vs number of observations for $t$-distributed data (with $p=50$, $\nu=5$).} \label{fig:MSE-covmat-t}
\end{figure}

\begin{figure}[!t]
\centerline{\includegraphics[width=0.5\textwidth]{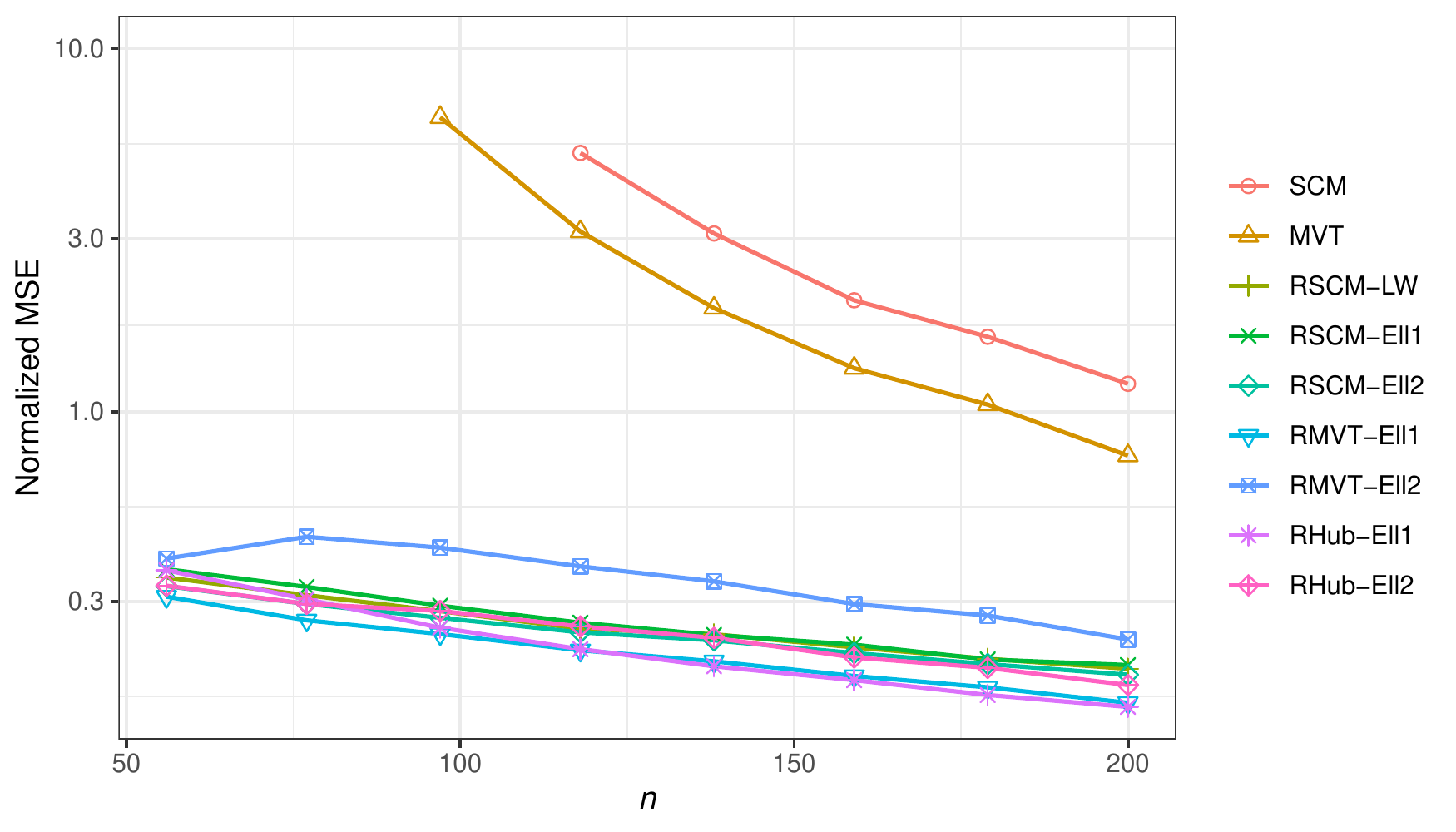}}
\caption{Normalized MSE of precision matrix vs number of observations for $t$-distributed data (with $p=50$, $\nu=5$).} \label{fig:MSE-precmat-t}
\end{figure}

\subsection{Results in terms of portfolio Sharpe ratio} \label{subsec:SR} 

After having shown the superiority of our proposed estimators, RMVT and RHub in terms of MSE in the estimation of the covariance matrix under heavy-tailed data, we now turn to assess the effects in terms of portfolio design. Note that an improvement of MSE in the covariance matrix may or may not translate into a significant improvement in terms of portfolio design; this depends on exactly what portfolio design is used and how it employs the estimated covariance matrix.

For simplicity, we consider the most basic Markowitz portfolio design \cite{Markowitz1952}:
$$\begin{array}{ll}
\underset{\mathbf{w}}{\textsf{minimize}} & \mathbf{w}^\top\boldsymbol{\Sigma}\mathbf{w}\\
\textsf{subject to} & \mathbf{w}^\top\boldsymbol{\mu} \ge \alpha\\
& \mathbf{1}^\top\mathbf{w}=1,
\end{array}$$
where $\boldsymbol{\mu}$ is the expected return of the returns and $\alpha$ is the minimum return desired for the portfolio.

We perform our backtest during the market period 2016-01-01 to 2020-01-31 on a rolling-window basis with a window length of 378 days (1.5 years). To make sure that our results are realistic, rather than performing a single backtest, we use the R package portfolioBacktest \cite{portfolioBacktest2019} to randomly select a large number of 200 datasets from the market data in the following way: each dataset chooses randomly $p=200$ stocks from the universe of 500, as well as a random period of 504 days (2 years) among the available period from 2016-01-01 to 2020-01-31.

\autoref{fig:MVP-SR} shows a boxplot with the Sharpe ratio\footnote{The Sharpe ratio is defined as the expected return normalized with the volatility or standard deviation: $\textsf{SR}=\frac{\mathbf{w}^\top\boldsymbol{\mu} - r_f}{\sqrt{\mathbf{w}^\top\boldsymbol{\Sigma}\mathbf{w}}}$, where $r_f$ is the return of the risk-free asset.} obtained mean-variance portfolio according to different estimators for the covariance matrix (along with two benchmarks: the index and the $1/N$ or uniform portfolio). We can observe that the two methods without shrinkage underperform the shrinkage methods (in particular the SCM). Among the shrinkage methods, we can see that our robust estimators slightly outperform the others, although the improvement is not extremely significative.

\begin{figure}[!t]
\centerline{\includegraphics[width=0.5\textwidth]{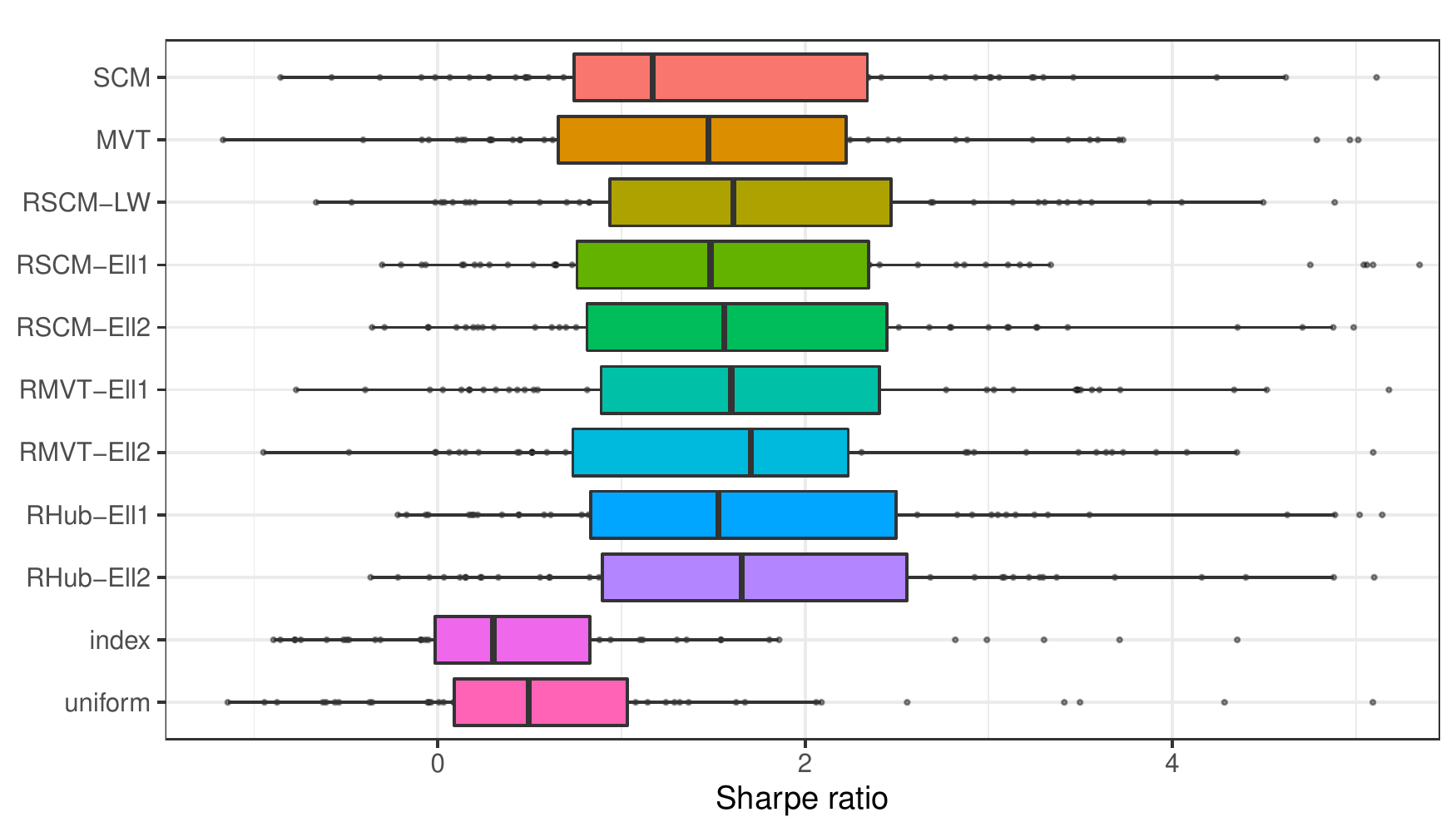}}
\caption{Boxplot of Sharpe ratio obtained by the mean-variance portfolio according to different estimators for the covariance matrix.} \label{fig:MVP-SR}
\end{figure}

\subsection{Supplementary studies} 

Supplementary material to this paper also contain comparison of the proposed methods  in the set-up described in \cite{ollila2019optimal}, where the global mean variance portfolio (GMVP)  is used as portfolio optimization strategy and the net returns correspond to $p=50$ stocks that are currently included in the Hang Seng Index (HSI).  
Compared methods include GMVP weight vector based on LW estimator and an estimator proposed in \cite{yang2015robust} that uses robust regularized Tyler's M-estimator with a tuning parameter selection  optimized for GMVP strategy.  We observed that GMVP based on the  proposed  RHub and RSCM are  the best performing methods in terms of realized risk.  We note that method of \cite{yang2015robust} was excluded in the study in \autoref{subsec:SR}  due to its high computational cost  in the high-dimensional setting.

%%%%%%%%%
% SECTION 7
%%%%%%%%%

\section{Conclusions and perspectives} \label{sec:concl}

This work proposed an original and fully automatic approach to compute an optimal shrinkage parameter in the context of heavy-tailed distributions and/or in presence of outliers. It has been shown that the performance of the method is similar to optimal one when the data is Gaussian while it outperforms shrinkage Gaussian-based methods when the data distribution turns out to be non-Gaussian. One of the benefits of the proposed adaptive shrinkage parameter selection is that it permits using real-valued or complex-valued data.  Furthermore, a MATLAB toolbox called ShrinkM is freely available at \url{http://users.spa.aalto.fi/esollila/shrinkM/} that includes functions to compute all of the proposed estimators (RHub, RTyl, RSCM, RMVT, and CV) as well as a script of one the simulation studies presented in the paper to reproduce the results. 
Furthermore, this paper opens several ways, notably considering the challenging cases where $p>n$ which is left for future work. Supplementary materials  also provide additional examples illustrating the benefits of the proposed estimators. 

\appendix

\subsection{Proof of \autoref{th:beta0}}\label{app:th:beta0}

Write $L(\be) = \MSE(\bo C_\be) = \E[\| \bo C_\be -\Mn\|_{\Fr}^2]$. Then note that 
\begin{align}
L(\be)
&= \E \big[  \big \| \be \bo C + (1- \be) p^{-1} \tr(\bo C) \I - \Mn \big\|^2_{\Fr} \big] \notag  \\ 
&= \E \Big[  \big \| \be(\bo C - \Mn)  + (1- \be) \big ( p^{-1} \tr(\bo C)  \I - \Mn \big) \big\|^2_{\Fr} \Big] \notag  \\   
&= \be^2 a_1 +  (1- \be)^2  a_2 + 2\be(1-\be) a_3 \label{eq:Th1_Lbe}
\end{align}
where 
 $a_1 =\E \big[\big\| \bo C - \Mn \big\|_{\Fr}^2] =  \E \big[ \tr(\bo C^2) \big] - \tr(\Mn^2)$, and
 \begin{align*}
a_2 &= \E \big [ \big\|  p^{-1} \tr(\bo C) \I  - \Mn  \big\|^2_{\Fr} \big] \\
&= a_3+ \tr(\Mn^2) -  p \eta_o^2  = a_3 + p ( \gamma -1)\eta_o^2  \\ 
a_3 &= p^{-1} \E \big[  \tr(\bo C) \tr(\bo C - \Mn) \big]  
= p^{-1}\E\big[	\tr (\bo C)^2  \big]  - \eta_o^2 p
\end{align*} 
and $\eta_o = \tr(\Mn)/p$. 
Note that $L(\be)$ is a convex quadratic function in $\be$ with a unique minimum given by 
\beq \label{eq:Th1_be0}
\be_o^\mathrm{app} =  \frac{a_2-a_3}{ (a_1-a_3) + (a_2-a_3)}.  
\eeq 
Substituting the expressions for constants $a_1, a_2$ and $a_3$ into $\be_o^\mathrm{app} $ yields the stated result.  

The expression for MSE of $\bo C_{\be_o^\mathrm{app}}$ then follows by substituting $\be_o^\mathrm{app}$ into expression for $L(\be)$ in \eqref{eq:Th1_Lbe} 
and using the relation, $(1-\be_o^\mathrm{app})  (a_2-a_3) = \be_o^\mathrm{app}(a_1-a_3)$,  
which follows from \eqref{eq:Th1_be0}. This yields the MSE expression 
\[
L(\be_o^\mathrm{app}) = a_2 - \be_o^\mathrm{app}(a_2-a_3)  =  a_3 + (1- \be_o^\mathrm{app}) (a_2-a_3).
\]
This gives the stated MSE expression after  noting that  $a_2-a_3 = \| \Mn - \eta_0 \bo I \|_{\Fr}^2$.

\subsection{Proof of \autoref{lem:EtrC}} \label{app:lem:EtrC}

First recall  that  $\Mn = \sca \M$, and hence
\begin{align} 
\Mor &=  \frac 1 n \sum_{i=1}^n u(\x_i^\top \Mn^{-1} \x_i ) \x_i \x_i^\top  \notag  \\ 
 	&= \Mn^{1/2}  \bigg \{  \frac 1 n \sum_{i=1}^n u\Big(\frac{r^2_i}{\sca}\Big) \frac{r_i^2}{\sca}  \u_i \u_i^\top  \bigg\} \Mn^{1/2}  ,  \label{eq:C_apu}
\end{align} 
where $\u_i = \M^{-1/2} \x_i/\|\M^{-1/2} \x_i\| $ and $r_i^2 = \|\M^{-1/2} \x_i\|^2$.   
Recall that stochastic representation theorem of elliptical random vectors states that $r_i$ is independent of $\u_i$ and $\u_i$-s 
are i.i.d. on a uniform distribution on the  unit sphere $\mathcal S^{p-1}=\{ \u \, : \,  \u^\top \u = 1\}$. 
Then note that 
\begin{align} \label{eq:EtraceC^2}
 &\E  \big[\tr \! \big(\Mor^2\big) \big]  \notag \\
 &= \frac{1}{\ndim^2} \E \bigg[\tr \!  \bigg\{  \Mn^{1/2}  \sum_{i=1}^\ndim u \Big( \frac{r^2_i}{\sca}\Big) \frac{r^2 _i}{\sca} \u_i \u_i^\top  \Mn^{1/2} \notag  \\ 
 &\qquad  \cdot \Mn^{1/2}  \sum_{j=1}^\ndim u \Big( \frac{r^2_j}{\sca}\Big) \frac{r^2_j}{\sca} \u_j \u_j^\top  \Mn^{1/2}  \bigg\} \bigg]  \notag \\ 
 &=    \frac{ 1 }{\ndim^2} \sum_{i=1}^\ndim \sum_{j=1}^\ndim    \E \bigg[   u \Big( \frac{r^2_i}{\sca}\Big) \frac{r^2 _i}{\sca}  \,  u\Big( \frac{r^2_j}{\sca}\Big)   \frac{r^2 _j}{\sca} \bigg]  \notag \\ 
& \qquad \times \E \bigg[ \underbrace{ \tr \!   \bigg\{   \u_i \u_i^\top  \Mn  \u_j \u_j^\top   \Mn  \bigg\}  }_{  = ( \u_i^\top \Mn \u_j)^2} \bigg]   \notag  \\ 
 &=    \frac{ 1 }{\ndim^2} \sum_{i=1}^\ndim  \E \bigg[ u \Big( \frac{r^2_i}{\sca}\Big)^2 \, \frac{r^4 _i}{\sca^2} \bigg]   \E \big[   ( \u_i^\top \Mn \u_i)^2 \big]  
\notag \\ &\qquad + 
 \frac{ 1 }{\ndim^2} \sum_{i \neq j }\E \bigg[  u \Big( \frac{r^2_i}{\sca} \Big) \frac{r^2 _i}{\sca} \bigg]     \E \bigg[  u\Big( \frac{r^2_j}{\sca} \Big) \frac{r^2_j}{\sca}  \bigg] \E \big[   ( \u_i^\top \Mn \u_j)^2 \big] \   \notag  \\
 &=    \frac{ 1 }{\ndim}  \E \bigg[   u \Big( \frac{r^2_1}{\sca}\Big)^2  \, \frac{r^4_1}{\sca^2} \bigg]   \E \big[   ( \u^\top_1 \Mn \u_1)^2 \big] \notag \\ 
 & \qquad  +  \Big( 1- \frac{1}{\ndim} \Big)   \bigg(\E \bigg[  u \Big( \frac{r^2_1}{\sca}\Big) \, \frac{r^2_1}{\sca} \bigg] \bigg)^2  \,   \E \big[   ( \u^\top_1 \Mn \u_2)^2 \big] .
 \end{align} 
In the second identity, we used the fact that $r_i$ is independent of $\u_i$ and $\tr(\A \B) = \tr(\B \A)$. In the 3rd identity we used  that $r_i$ is independent of $r_j$ and in the 4th identity, we used that $\u_i$-s and $r_i$-s are i.i.d. 
 
Note that  $\E \big[  u(r^2_1/\sca) (r^2_1/\sca) \big] = \pdim$  due to \eqref{eq:sca} and 
 $\E \big[   u(r^2_1/\sca)^2 (r^4_1/\sca^2) \big] = \psi_1 \pdim (\pdim+2)$. 
 Using these facts and the following results from \cite{chen2010shrinkage}:
 \begin{align} 
\E \big[   ( \u^\top_1 \Mn \u_1)^2 \big]  &= \frac{ 2 \tr(\Mn^2) + \tr(\Mn)^2}{p(p+2)}  \label{eq:Expect1} \\
\E \big[   ( \u^\top_1 \Mn \u_2)^2 \big]  &=    \frac{\tr(\Mn^2)}{\pdim^2}  ,  \label{eq:Expect2}
\end{align} 
 we get  
 \begin{align*} 
&\E\big[\tr \! \big(\Mor^2\big) \big]   \\
&=    \frac{ \psi_1  }{\ndim}   \big( 2 \tr(\Mn^2) + \tr(\Mn)^2 \ \big)  + 
 \Big( 1- \frac{1}{\ndim} \Big)    \tr(\Mn^2)
 \\ 
&=\left( 1 + \frac{2 \psi_1-1}{\ndim} \right) \tr(\Mn^2) +  \frac{\psi_1}{\ndim}\tr(\Mn)^2  .
\end{align*} 
Next note that 
\begin{align*}
\tr(\bo C)^2 &= \left\{ \frac 1 n \sum_{i=1}^n  u \Big( \frac{r^2_i}{\sca}\Big)   \frac{r^2_i}{\sca} \u_i^\top \Mn \u_i \right\}^2 \\
&= \frac{ 1}{n^2} \sum_{i=1}^n \sum_{j=1}^n   u \Big( \frac{r^2_i}{\sca}\Big)   \frac{r^2_i}{\sca}  u \Big( \frac{r^2_j}{\sca}\Big)   \frac{r^2_j}{\sca} \u_i^\top \Mn \u_i \u_j^\top \Mn \u_j
\end{align*}
and hence
\begin{align*}
\E[\tr(\bo C)^2] &=\frac{1}{n^2} \sum_{i=1}^n \E\bigg[   u \Big( \frac{r^2_i}{\sca}\Big)^2  \frac{r^4_i}{\sca^2}  \bigg] \E \big[ ( \u_i^\top \Mn \u_i)^2 \big]  \\ 
&\ \ + \frac{ 1}{n^2} \sum_{i \neq j}   \E \bigg[ u \Big( \frac{r^2_i}{\sca}\Big)  \frac{r^2_i}{\sca} \bigg] \E \bigg[  u \Big( \frac{r^2_j}{\sca}\Big)   \frac{r^2_j}{\sca}  \bigg] \frac{\tr(\Mn)^2}{p^2} \\
&= \frac{\psi_1 p (p+2)}{n} \E \big[ ( \u_i^\top \Mn \u_i)^2 \big]  + \Big( 1- \frac 1 n  \Big) \tr(\Mn)^2.
\end{align*}
In the first identity  we used that 
\[
\E[\u_i^\top \Mn \u_i ] = \tr\{ \E[\u_i \u_i^\top] \Mn \} = \tr(\Mn)/p
\]
as $\E[\u_i \u_i^\top] = (1/p) \bo I$ and the fact that $r_i$ is independent of $\u_i$. 
In the second identity we used that samples are i.i.d. and  $\E \big[  u(r^2_i/\sca) (r^2_i/\sca) \big] = \pdim$  due to \eqref{eq:sca} and 
 $\E \big[   u(r^2_i/\sca)^2 (r^4_i/\sca^2) \big] = \psi_1 \pdim (\pdim+2)$. The result then follows by substituting \eqref{eq:Expect1} into the last equation.

\subsection{Proof of \autoref{lem:EtrC_complex}} \label{app:lem:EtrC_complex}

%\noindent {\bf Proof of Lemma~\ref{lem:EtrC}}.  
Using  $\Mn = \sca \M$ and  \eqref{eq:C_apu} 
one obtains as in \autoref{lem:EtrC} the following expression 
 \begin{align} \label{eq:EtraceC^2}
 &\E  \big[\tr \! \big(\Mor^2\big) \big]  \notag \\
  &=    \frac{ 1 }{\ndim}  \E \bigg[   u \Big( \frac{r^2_1}{\sca}\Big)^2  \, \frac{r^4_1}{\sca^2} \bigg]   \E \big[   ( \u^\hop_1 \Mn \u_1)^2 \big] \notag \\ 
 & \qquad  +  \Big( 1- \frac{1}{\ndim} \Big)   \bigg(\E \bigg[  u \Big( \frac{r^2_1}{\sca}\Big) \, \frac{r^2_1}{\sca} \bigg] \bigg)^2  \,   \E \big[   |\u^\hop_1 \Mn \u_2|^2 \big] .
 \end{align} 
Using the facts that  $\E \big[  u(r^2_1/\sca) (r^2/\sca) \big] = \pdim$  due to \eqref{eq:sca} and 
 $\E \big[   u(r^2/\sca)^2 (r^4/\sca^2) \big] = \psi_1 \pdim (\pdim+1)$  
 along with the facts that  \cite[cf. eq. (66), (67)] {chen2011robust}
\begin{align}
\E \big[   ( \u^\hop \M \u)^2 \big]  &=  [p(p+1)]^{-1} ( \tr(\M^2) + \tr(\M)^2) , \label{eq:lemma2_apu1} \\ 
\E \big[   ( \u^\hop_1 \M \u_2)^2 \big]  &=  \pdim^{-2} \tr(\M^2) ,\label{eq:lemma2_apu2}
\end{align} 
 we get  
 \begin{align*} 
&\E\big[\tr \! \big(\Mor^2\big) \big]   \\
&=    \frac{ \psi_1  }{\ndim}   \big( \tr(\Mn^2) + \tr(\Mn)^2 \ \big)  + 
 \Big( 1- \frac{1}{\ndim} \Big)    \tr(\Mn^2)
 \\ 
&=\left( 1 + \frac{\psi_1-1}{\ndim} \right) \tr(\Mn^2) +  \frac{\psi_1}{\ndim}\tr(\Mn)^2  .
\end{align*} 
Using similar proof as in proof of \autoref{lem:EtrC}, we obtain 
\begin{align} \label{eq:lemma2_apu3}
\E[\tr(\bo C)^2] = \frac{\psi_1 p (p+1)}{n} \E \big[ ( \u_i^\hop \Mn \u_i)^2 \big]  + \Big( 1- \frac 1 n  \Big) \tr(\Mn)^2.
\end{align}
The result then follows by substituting the expression from \eqref{eq:lemma2_apu2} into the last equation.

%\hspace*{\fill} 
 %$ \blacksquare$

\subsection{Proof of \autoref{lem:kappa}} \label{app:lem:kappa}

We show the result in the real case only as the result follows similarly for complex-valued case. 
Write $ \bo z = \M^{-1/2} \x$ and note that $ \bo z \sim \mathcal E_p(\bo 0, \bo I,g)$. 
The result 
\beq \label{eq:psi1_SCM}
\kappa =   \frac{\E[\| \bo z \|^4]}{p(p+2)} -1  =  \frac{1}{3} \big( \E[z_1^4] - 3) 
\eeq
follows by recalling the stochastic decomposition. Namely,  $\bo z =_d r \u$, where $r =_d  \| \bo z \|$ is independent of $\u$,  and $\u$ possesses a uniform distribution on the unit sphere $\mathcal S^{p-1}$. Thus  we have that 
\[
\E[z_i^4] =  \E[ \| \bo z \|^4] \E[v_i^4] = 3 \psi_1 
\]
where we used that $ \E[ v_i^4] = 3 (p(p+2))^{-1}$ (see e.g. \cite[Lemma A.1.]{ollila2003affine} and that $\psi_1 = \E[ \| \bo z \|^4]/p(p+2)$.  Furthermore, since $\E[z_i^2] = 1$, \eqref{eq:psi1_SCM} states that  $ \kappa = (1/3) \mathsf{kurt}(z_i)$.  Then since $\mathsf{kurt}(z_i)=\mathsf{kurt}(x_i)$,  we have the stated result that  $ \kappa = (1/3) \mathsf{kurt}(x_i)$.

%\section*{Acknowledgment}
%The authors would like to thank the reviewers for insightful comments which
%helped to improve the paper.

\ifCLASSOPTIONcaptionsoff
\newpage
\fi




\bibliographystyle{IEEEtran}
%\bibliography{shrinkMbib}


%\begin{IEEEbiography}{Esa Ollila}
%\end{IEEEbiography}

%\begin{IEEEbiography}{Dani}
%\end{IEEEbiography}

%\begin{IEEEbiography}{Frederico}
%\end{IEEEbiography}

%\newpage

%\begin{IEEEbiographynophoto}{Esa Ollila}
%Biography text here.s


%\vfill



\clearpage 
\end{document}